\documentclass[11pt]{article}	

\usepackage{palatino}
\usepackage{braket}
\usepackage{amsfonts}
\usepackage{amssymb}
\usepackage{amsmath}
\usepackage{latexsym}
\usepackage{amsthm}
\usepackage[usenames,dvipsnames]{color}
\usepackage{soul}
\usepackage{hyperref}
\usepackage{paralist}
\usepackage{fullpage}
\usepackage{comment}
\usepackage{mathdots}

\hypersetup{pdfpagemode=UseNone}

\makeatletter
\newtheorem*{rep@theorem}{\rep@title}
\newcommand{\newreptheorem}[2]{\newenvironment{rep#1}[1]{\def\rep@title{#2 \ref{##1}}\begin{rep@theorem}}{\end{rep@theorem}}}
\makeatother

\newtheorem{theorem}{Theorem}[section]
\newreptheorem{theorem}{Theorem}
\newtheorem{lemma}[theorem]{Lemma}

\newtheorem{corollary}[theorem]{Corollary}
\newtheorem{definition}[theorem]{Definition}

\newtheorem{claim}[theorem]{Claim}

\theoremstyle{definition}
{
	\newtheorem{remark}[theorem]{Remark}
	\newtheorem{obs}[theorem]{Observation}
}

\newcommand{\snorm}[1]{\norm{#1}_{\mathrm {\infty}}}    \renewcommand{\comment}[1]{}

\newcommand{\tinyspace}{\mspace{1mu}}

\newcommand{\abs}[1]{\left\lvert\tinyspace #1 \tinyspace\right\rvert}

\newcommand{\norm}[1]{\left\lVert\tinyspace #1 \tinyspace\right\rVert}
\newcommand{\tr}{\operatorname{Tr}}
\newcommand{\pr}{\operatorname{Pr}}

\newcommand{\class}[1]{\textup{#1}}

\newcommand{\ayes}{A_{\rm yes}}
\newcommand{\ano}{A_{\rm no}}

\mathchardef\mhyphen="2D

\newcommand{\inner}[2]{\langle #1, #2 \rangle}

\newcommand{\calC}{\mathcal{C}}

\newcommand{\calL}{\mathcal{L}}

\newcommand{\ketbra}[2]{\ket{#1}\!\bra{#2}}

\def\complex{\mathbb{C}}
\def\real{\mathbb{R}}
\def\natural{\mathbb{N}}

\def\({\left(}
\def\){\right)}

\def\yes{\text{yes}}
\def\no{\text{no}}

\newcommand{\QMA}{\class{QMA}}
\newcommand{\QMAt}{\class{QMA}(2)}
\newcommand{\PSPACE}{\class{PSPACE}}
\newcommand{\BQP}{\class{BQP}}
\newcommand{\QCMA}{\class{QCMA}}

\newcommand{\PrQCMA}{\class{Precise}\text{-}\QCMA}
\newcommand{\PrBQP}{\class{Precise}\text{-}\BQP}

\newcommand{\BQPC}{\class{QCIRCUIT}}
\newcommand{\DDET}{\class{NON-EMPTY GAP}}

\newcommand{\EXP}{\class{EXP}}
\newcommand{\coEXP}{\class{co-EXP}}

\newcommand{\NEXP}{\class{NEXP}}

\newcommand{\QCSi}{\mathrm{QC} \Sigma}

\newcommand{\QSi}{\mathrm{Q} \Sigma}
\newcommand{\QCPi}{\mathrm{QC} \Pi}
\newcommand{\QPi}{\mathrm{Q} \Pi}
\newcommand{\QCPH}{\class{QCPH}}
\newcommand{\QPH}{\class{QPH}}
\newcommand{\poly}{\textup{poly}}

\newcommand{\mpoly}{\textup{mpoly}}

\newcommand{\sigmatwo}{\Sigma_2}
\newcommand{\sigmai}{\Sigma_i}
\newcommand{\pit}{\Pi_2}
\newcommand{\pii}{\Pi_i}

\newcommand{\PH}{\class{PH}}
\newcommand{\PP}{\class{PP}}
\newcommand{\MA}{\class{MA}}

\newcommand{\Po}{\class{P}}
\newcommand{\NP}{\class{NP}}

\newcommand{\pp}{\class{PP}}
\newcommand{\cp}{\#\Po}

\newcommand{\BPP}{\class{BPP}}
\newcommand{\Ppoly}{\class{P$_{/\poly}$}}
\newcommand{\Bpoly}{\class{BQP$_{/\mpoly}$}}

\newcommand{\coNP}{\class{co}\text{-}\NP}

\newcommand{\QRGone}{\class{QRG}(1)}

\usepackage{graphicx}
\definecolor{greenn}{rgb}{0,0.8,0.2}
\definecolor{bluue}{rgb}{0.3,0,0.7}

\begin{document}
\title{\bf Quantum generalizations of the polynomial hierarchy \\with applications to $\class{QMA(2)}$}

\author{
	Sevag Gharibian\thanks{University of Paderborn, Paderborn, North Rhine-Westphalia, Germany, and Virginia Commonwealth University, Richmond, VA, USA.}
	\and
	Miklos Santha\thanks{CNRS, IRIF, Universit\'{e} de Paris, France and Centre for Quantum Technologies, National University of Singapore, Singapore.}
	\and
	Jamie Sikora\thanks{Perimeter Institute for Theoretical Physics, Waterloo, Ontario, Canada.}
	\and
	Aarthi Sundaram\thanks{Joint Center for Quantum Information and Computer Science, University of Maryland, College Park, Maryland, USA.}
	\and
	Justin Yirka\thanks{The University of Texas at Austin, Austin, Texas, USA}
}

\date{}

\maketitle

\vspace{-5mm}

\begin{abstract}
The polynomial-time hierarchy (PH) has proven to be a powerful tool for providing separations in computational complexity theory (modulo standard conjectures such as PH does not collapse). Here, we study whether two quantum generalizations of PH can similarly prove separations in the quantum setting. The first generalization, $\QCPH$, uses classical proofs, and the second, $\QPH$, uses quantum proofs. For the former, we show quantum variants of the Karp-Lipton theorem and Toda's theorem. For the latter, we place its third level, $\QSi_3$, into $\class{NEXP}$ {using the Ellipsoid Method for efficiently solving semidefinite programs}. These results yield two implications for $\class{QMA(2)}$, the variant of Quantum Merlin-Arthur (QMA) with two unentangled proofs, a complexity class whose characterization has proven difficult. First, if $\QCPH=\QPH$ (i.e., alternating quantifiers are sufficiently powerful so as to make classical and quantum proofs ``equivalent''), then $\class{QMA(2)}$ is in the Counting Hierarchy (specifically, in $\Po^{\pp^{\pp}}$). Second, unless $\class{QMA(2)}={\QSi_3}$ (i.e., alternating quantifiers do not help in the presence of ``unentanglement''), $\QMA(2)$ is strictly contained in NEXP.
\end{abstract}

\section{Introduction}
\label{sec:intro}

The polynomial-time hierarchy ($\PH$)~\cite{MS72} is a staple of computational complexity theory, and generalizes $\Po$, $\NP$ and $\coNP$ with the use of alternating existential ($\exists$) and {universal} ($\forall$) operators.
Roughly, a language $L\subseteq\set{0,1}^*$ is in $\sigmai$, the $i${th} level of PH, if there exists a polynomial-time deterministic Turing machine $M$ that acts as a verifier and accepts $i$ proofs $y_1,\ldots, y_i$, each polynomially bounded in the length of the input $x$, such that:
	\begin{eqnarray}
    x\in L&\Rightarrow&\exists y_1\forall y_2\exists y_3\cdots Q_i y_i\text{ such that $M$ accepts }(x,y_1,\ldots, y_i), \\
    x\not\in L&\Rightarrow&\forall y_1\exists y_2\forall y_3\cdots \overline{Q}_i y_i\text{ such that $M$ rejects }(x,y_1,\ldots, y_i),
\end{eqnarray}
where $Q_i=\exists$ if $i$ is odd and $Q_i=\forall$ if $i$ is even, and $\overline{Q}$ denotes the complement of $Q$.
Then, \PH\ is defined as the union over all $\sigmai$ {for all $i \in \mathbb{N}$}.
The study of PH has proven remarkably fruitful in the classical setting, from celebrated results such as Toda's Theorem~\cite{T91}, which shows that $\class{PH}$ is contained in $\Po^{\#\Po}$, to the Karp-Lipton Theorem~\cite{KL80}, which says that unless PH collapses to its second level, \NP\ does not have polynomial size non-uniform circuits.

As PH has played a role in separating complexity classes (assuming standard conjectures like "PH does not collapse"), it is natural to ask whether {quantum} generalizations of PH can be used to separate \emph{quantum} complexity classes. Here, there is some flexibility in defining ``quantum PH'', as there is more than one well-defined notion of ``quantum NP'': The first, Quantum-Classical Merlin Arthur (\QCMA)~\cite{AN02}, is a quantum analogue of Merlin-Arthur (MA) with a classical proof but quantum verifier. The second, Quantum Merlin Arthur (\QMA)~\cite{KSV02}, is \QCMA\ except with a quantum proof. Generalizing each of these definitions leads to (at least) two possible definitions for ``quantum PH'', the first using classical proofs (denoted \QCPH), and the second using quantum proofs (denoted \QPH) (formal definitions in Section~\ref{sec:preliminaries}).

With these definitions in hand, our aim is to separate quantum classes whose complexity characterization has generally been difficult to pin down. A prime example is $\QMAt$, the variant of QMA with two ``unentangled'' quantum provers. While the classical analogue of $\QMAt$ (i.e. an MA proof system with two provers) trivially equals MA, in the quantum regime multiple unentangled provers are conjectured to yield a more powerful proof system (e.g. there exist problems in $\QMAt$ not known to be in $\QMA$)~\cite{LCV07, BT10, Bei08, ABDFS09}. For this reason, $\QMAt$ has received much attention, despite which the strongest bounds known on $\QMAt$ remain the trivial ones:
\[
	\QMA\subseteq \QMAt\subseteq \NEXP.
\]
(Note: $\QMA\subseteq\PP$~\cite{KW00, Vy03, MW05, GY16}.) In this work, we show that, indeed, results about the structure of \QCPH\ or \QPH\ yield implications about the power of QMA(2).

\subsection{Results, techniques, and discussion}
\label{sec:results}
We begin by informally defining the two quantum generalizations of PH to be studied (formal definitions in Section~\ref{sec:preliminaries}).

\paragraph{How to define a ``quantum \PH''?} The first definition, $\QCPH$, has its $i$th level $\QCSi_i$ defined analogously to $\sigmai$, except we replace the  Turing machine $M$ with a polynomial-size uniformly generated quantum circuit $V$ such that:
\begin{eqnarray}
    x\in \ayes\!&\!\!\Rightarrow\!\!&\!\exists y_1\forall y_2\exists y_3\cdots Q_i y_i\text{ such that $V$ accepts }(x,y_1,\ldots, y_i)\text{ with probability $\geq 2/3$}, \label{eq:Ayes} \\
    x\in\ano\!&\!\!\Rightarrow\!\!&\!\forall y_1\exists y_2\forall y_3\cdots \overline{Q}_i y_i\text{ such that $V$ accepts }(x,y_1,\ldots, y_i)\text{ with probability $\leq 1/3$}, \label{eq:Ano}
\end{eqnarray}
where the use of a language $L$ has been replaced with a \emph{promise problem\footnote{Recall that unlike a decision problem, for a promise problem $A =(\ayes, \ano)$, it is not necessarily true that for all inputs $x\in\Sigma^*$, either $x\in\ayes$ or $x\in \ano$. In the case of proof systems such as \QCPH, when $x\not\in \ayes\cup\ano$, $V$ can output an arbitrary answer. Additionally, $\ayes \cap \ano = \emptyset$.}} $A=(\ayes,\ano)$ (since $\QCSi_i$ uses a bounded error verifier).  The values $(2/3, 1/3)$ are \emph{completeness} and \emph{soundness} parameters for $A$ and the interval $(1/3, 2/3)$ where no acceptance probabilities are present is termed the \emph{promise gap} for $A$. Notice that \QCPH\, defined as $\bigcup_{i \in \natural} \QCSi_i$, is a generalization of \QCMA\, in that $\QCSi_1=\QCMA$.

We next define \QPH\ using \emph{quantum} proofs. Here, however, there are various possible definitions one might consider. Can the quantum proofs be {entangled} between alternating quantifiers? If not, we are enforcing ``unentanglement'' as in $\class{QMA(2)}$. Allowing entanglement, on the other hand, might yield classes similar to \class{QIP}; however, note that $\class{QIP}=\class{QIP(3)}$ (i.e. $\class{QIP}$ collapses to a $3$-message proof system)~\cite{KW00,MW05}, and so it is not clear that allowing entanglement leads to an interesting hierarchy. Assuming proofs are unentangled, should the proofs be {pure} or {mixed} quantum states? (For QMA and $\QMAt$, standard convexity arguments show both classes of proofs are equivalent, but such arguments fail when \emph{alternating} quantifiers are allowed.)

Here, we define $\QPH$ to have its $i$th level, $\QSi_i$, defined similarly to $\QCSi_i$, except each classical proof $y_j$ is replaced with a mixed quantum state $\rho_j$ on polynomially many qubits (for clarity, each $\rho_j$ acts on a disjoint set of qubits). We say a promise problem $A = (\ayes, \ano)$ is in $\QSi_i$ if it satisfies the following conditions:
\begin{eqnarray*}
    x\in \ayes\!&\!\!\Rightarrow\!\!&\!\exists \rho_1\forall \rho_2\exists \rho_3\cdots Q_i \rho_i\text{ such that $V$ accepts }(x,\rho_1,\ldots, \rho_i)\text{ with probability $\geq 2/3$}, \\
    x\in\ano\!&\!\!\Rightarrow\!\!&\!\forall \rho_1\exists \rho_2\forall \rho_3\cdots \overline{Q}_i \rho_i\text{ such that $V$ accepts }(x,\rho_1,\ldots, \rho_i)\text{ with probability $\leq 1/3$}.
\end{eqnarray*}
\noindent Note that $\QMA=\QSi_1$ and $\QMAt \subseteq \QSi_3$ (simply ignore the second proof).

Our results are now stated as follows under three headings.

\paragraph{1. An analogue of Toda's theorem for $\QCPH$.}
As previously mentioned, \PH\ is one way to generalize $\NP$ using alternations. Another approach is to {count} the number of solutions for  an $\NP$-complete problem such as SAT, as captured by $\cp$. Surprisingly, these two notions are related, as shown by the following celebrated theorem of Toda.

\begin{theorem}[Toda's theorem~\cite{T91}]
\label{thm:Toda}
$\PH \subseteq \Po^{\#\Po}$.
\end{theorem}

\noindent In the quantum setting, for $\QCPH$, it can be shown using standard arguments involving enumeration over classical proofs that $\QCPH \subseteq \PSPACE$. However, here  we show a stronger result.

\begin{theorem}[A quantum-classical analogue of Toda's theorem]\label{thm:QCTODA}
$\QCPH \subseteq \Po^{\pp^{\pp}}$.
\end{theorem}
\noindent Thus, we {almost} recover the original bound of Toda's theorem\footnote{$\pp$ captures all problems for which a majority of all possible answers is correct and it is known that $\Po^{\pp} = \Po^{\#\Po}$.}, except we require an oracle for the \emph{second} level of the Counting Hierarchy (CH). CH can be defined with its first level as $\mathbf{C}_1 = \pp$ and its $k$th level for $k\geq 2$ as $\mathbf{C}_{k} = \pp^{\mathbf{C}_{k-1}}$.

{Why} did we move up to the next level of CH? There are two difficulties in dealing with \QCPH\ (see Section~\ref{sec:toda} for a detailed discussion). The first can be sketched as follows. Classically, many results involving PH, from basic ones implying the collapse of PH to more advanced statements such as Toda's theorem, use the following recursive idea (demonstrated with $\sigmatwo$ for simplicity): By fixing the existentially quantified proof of $\sigmatwo$ the remnant reduces to a $\class{co-NP}$ problem, i.e. we can recurse to a lower level of PH. In the quantum setting, however, this does not hold --- fixing the existentially quantified proof for $\QCSi_2$ does \emph{not} necessarily yield a $\class{co-QCMA}$ problem as some acceptance probabilities may fall in the $(1/3, 2/3)$ promise gap which cannot happen for a  problem in $\class{co-QCMA}$. (This is due to the same phenomenon that has been an obstacle to resolving whether $\exists\boldsymbol{\cdot}\BPP$ equals $\MA$ (see Section~\ref{sec:rel} and Remark~\ref{rem:langvsprom}).)
Thus, we cannot directly generalize recursive arguments from the classical setting to the quantum setting. The second difficulty is trickier to explain briefly (see Section~\ref{sscn:bounding} for details). Roughly, Toda's proof that $\PH \subseteq \Po^{\pp}$ crucially relies on the Valiant-Vazirani (VV) theorem~\cite{VV86}, which has one-sided error (i.e.~VV may map \textsf{YES} instances of SAT to \textsf{NO} instances of UNIQUE-SAT, but \textsf{NO} instances of SAT are always mapped to \textsf{NO} instances of UNIQUE-SAT). The VV theorem for QCMA~\cite{ABBS08} also has this property, but in addition it can output instances which are ``invalid''. Roughly, an ``invalid'' instance of a promise problem $\Pi$ is an instance violating the promise of $\Pi$. Such instances pose a problem, because feeding an oracle an invalid instance results in an arbitrary output; coupled with the two-sided error which arises in \QCPH\ due to the presence of {alternating} quantifiers, it is unclear how to extend the parity arguments used in Toda's proof to the \QCPH\ setting. 

To circumvent these difficulties, we exploit a high-level idea from~\cite{GY16}, where an oracle for SPECTRAL GAP\footnote{This problem determines whether the spectral gap of a given local Hamiltonian is ``small'' or ``large''.} was used to detect ``invalid'' QMA instances\footnote{This was used, in turn, to show in conjunction with~\cite{A14} that SPECTRAL GAP is $\Po^{\class{Unique-QMA[log]}}$-hard.}.
In our setting, the ``correct'' choice of oracle turns out to be a \PrBQP\ oracle\footnote{{For the purposes of our Cleaning Lemma, we may instead use a PQP oracle, where recall PQP is BQP except in the YES case, the verifier accepts with probability $>1/2$, and in the NO case accepts with probability $\leq 1/2$. Note that in contrast to \PrBQP, PQP is defined without completeness/soundness parameters~\cite{W09_2}; that one may impose an inverse exponential promise gap on PQP is a non-trivial consequence of the fact that one may choose an ``appropriately nice'' gate set for PQP to ensure acceptance probabilities are rational. For this reason, and since whether \PrBQP\ equals PQP or not depends strongly on the choice of completeness/soundness parameters, we have opted to treat \PrBQP\ as a generally distinct entity from PQP; see Section~\ref{sec:prbqp} for details.}}, where \PrBQP\ is roughly \BQP\ with an {inverse} exponentially small promise gap. Using this, we are able to essentially ``remove'' the promise gap of \QCPH\ altogether, thus recovering a ``decision problem'' which does not pose the difficulties above. Specifically, this mapping is achieved by Lemma~\ref{l:clean} (Cleaning Lemma), which shows that $\forall i \in \natural,$ we have $\QCSi_i\subseteq \exists\boldsymbol{\cdot}\forall\boldsymbol{\cdot}\cdots\boldsymbol{\cdot} Q_i\boldsymbol{\cdot} \class{P}^{\PP}$. {The latter expression applies the existential ($\exists$) and universal ($\forall$) operators to a complexity class $\calC$. Informally, $\exists \boldsymbol{\cdot} \calC$ is the class of languages such that an input $x$ is in the language if and only if there is a polynomial-sized witness $y$ such that $\langle x, y \rangle$ is in a language in $\calC$. Correspondingly, the $\forall \boldsymbol{\cdot} \calC$ class is defined when for every witness $y$, $\langle x, y \rangle$ is in some language in $\calC$. (See Definition~\ref{def:exists} for formal definitions of $\exists$ and $\forall$.)

Notice that although we use a \PrBQP\ oracle above, the Cleaning Lemma shows containment using a \PP\ oracle. This is because, as shown in Lemma~\ref{l:containedPP} and Corollary~\ref{cor:equalsPP}, $\PrBQP\subseteq \PP$. One may ask whether our proof technique also works with an oracle \emph{weaker} than \PP. We show in Theorem~\ref{thm:ppcomplete} that this is unlikely, since the problem of detecting proofs in promise gaps of quantum verifiers is \PP-complete.

Finally, an immediate corollary of Theorem~\ref{thm:QCTODA} and the fact that $\QMAt\subseteq \QPH$ is:
\begin{corollary}\label{cor:QMA2_1}
    If $\QCPH=\QPH$, then $\QMAt\subseteq \Po^{\PP^{\PP}}$.
\end{corollary}
\noindent In other words, if alternating quantifiers are so powerful so as to make classical and quantum proofs equivalent in power, then $\QMAt$ is contained in CH (and thus in \PSPACE). For comparison, $\QMA\subseteq \Po^{\class{QMA[log]}}\subseteq \PP$~\cite{KW00,Vy03,MW05,GY16}.

\paragraph{2. \QPH\ versus \NEXP.} We next turn to the study of \emph{quantum} proofs, i.e.~\QPH. As mentioned above, the best known upper bound on $\QMAt$ is \NEXP\ --- a non-deterministic verifier can simply guess an exponential-size description of the proof. When alternating quantifiers are present, however, this strategy seemingly no longer works. In other words, it is not even clear that $\QPH \subseteq\NEXP$! This is in stark contrast to the explicit $\Po^{\cp}$ upper bound for $\PH$~\cite{T91}. In this part, our goal is to use semidefinite programming to give bounds on some levels of \QPH. As we will see, this will yield the existence of a complexity class lying ``between'' $\QMAt$ and \NEXP.

\begin{theorem}[Informal statement]
\label{thm:infEXP}
It holds that $\QSi_2 \subseteq \EXP$ and
$\QPi_2 \subseteq \EXP$, even when the completeness-soundness gap is inverse doubly-exponentially small.
\end{theorem}

\noindent The proof idea is to map alternating quantifiers to an optimization problem with alternating minimizations and maximizations.
Namely, to decide if $x \in A_{\yes}$ or $x \in A_{\no}$ for a $\QSi_i$ promise problem $A=(\ayes,\ano)$, where $i$ is even, we can solve for $\alpha$ defined as the optimal value of the optimization problem:
\begin{equation}
\alpha := \max_{\rho_1} \min_{\rho_2} \max_{\rho_3} \cdots \min_{\rho_i}\;
\inner{C}{\rho_1 \otimes \rho_2 \otimes \cdots \otimes \rho_i} \label{eq:alpha1}
\end{equation}
where $C$ is the POVM operator\footnote{A POVM is a set of Hermitian positive semidefinite operators that sum to the identity. In this case, the POVM has two operators --- corresponding to the \textsf{ACCEPT} and \textsf{REJECT} states of the verifier.} corresponding to the  \textsf{ACCEPT} state of the verifier.
{(Note that each optimization attains an optimal solution by a simple compactness/continuity argument, hence the use of ``max'' instead of ``sup'' and ``min'' instead of ``inf''.)}
This is a non-convex problem, and as such is (likely) hard to solve in general.
Our approach is to cast the case of $i=2$ as a semidefinite program (SDP), allowing us to {efficiently} approximate $\alpha$.

The next natural question is whether a similar SDP reformulation might be used to show whether $\class{Q$\Sigma$}_3$ or $\class{Q$\Pi$}_3$ is also contained in \EXP. Unfortunately, this is likely to be difficult --- indeed, if there exists a ``nice'' SDP for the optimal success probability of $\QSi_3$ protocols, then it would imply $\QMAt \subseteq \EXP$, resolving the longstanding open problem of separating $\QMAt$ from $\NEXP$ (recall $\QMAt\subseteq\QSi_3$).
Likewise, a ``nice'' SDP for $\class{Q$\Pi$}_3$ would place $\class{co-}\QMAt \subseteq \EXP$.

To overcome this, we resort to non-determinism by stepping up to NEXP. Namely, one can non-deterministically guess the first proof of a
$\QSi_3$ protocol, then approximately solve the SDP for the resulting $\QPi_2$-flavoured computation. Hence, we have:

\begin{theorem}[Informal Statement] \label{IntroThm:NEXP}
It holds true that $\QMA(2) \subseteq \QSi_3 \subseteq \NEXP$ and $\mathrm{co} \textrm{-} \QMA(2) \subseteq \QPi_3 \subseteq {\mathrm{co} \textrm{-} \NEXP}$,
even when the completeness-soundness gap is inverse doubly-exponentially small. All containments hold with equality in the inverse exponentially small completeness-soundness gap setting as $\QMA(2) = \NEXP$ in this case~\textup{\cite{attila}}.
\end{theorem}

\noindent Three remarks are in order. First, note that our results are independent of the gate set. Second, in principle, it remains plausible that the \emph{fourth} level of \QPH\ already exceeds \NEXP\ in power. Finally, we have the following implication for $\QMAt$. Assuming \PH\ does not collapse, alternating quantifiers strictly add power to NP proof systems. If alternating quantifiers similarly add power in the {quantum} setting, then it would separate \QMAt\ from \NEXP\ via the following immediate corollary.
\begin{corollary}
If $\QMAt \neq \QSi_3$, i.e. if the second universally quantified proof of $\QSi_3$ adds proving power, then $\QMAt \neq \NEXP$.
Similarly, if $\mathrm{co} \textrm{-} \QMA(2) \neq \QPi_3$, then $\mathrm{co} \textrm{-} \QMA(2) \neq \mathrm{co} \textrm{-} \NEXP$.
\end{corollary}

\noindent{\textbf{Note added:} Since the original release of this article, new observations relevant to the discussion above have been made. Since these observations are closely related to the open questions of this article, they have been placed in Section~\ref{sscn:open} (Recent observations and open questions).}

\paragraph{3. A quantum generalization of the Karp-Lipton Theorem.}
Finally, our last result studies a topic which is unrelated to QMA(2) --- the well-known Karp-Lipton theorem \cite{KL80}. The latter shows that if $\NP$-complete problems can be solved by polynomial-size non-uniform Boolean circuits, then $\sigmatwo = \pit$ (formal definitions in Section~\ref{sec:preliminaries}), which in turn implies that $\PH$ collapses to its second level. {Here, a poly-size ``non-uniform'' circuit receives, in addition to the input instance, a poly-size ``advice string'' $y$ such that (1) $y$ depends only on the input size $n$, and (2) given $n$, computing $y$ need not be poly-time.} The class of decision problems solved by such circuits is $\Ppoly$.

\begin{theorem}[Karp-Lipton~\cite{KL80}]
If $\class{NP} \subseteq \Ppoly$ then $\pit = \sigmatwo$.
\end{theorem}
\noindent Denote the bounded-error analogue of $\Ppoly$ with polynomial-size non-uniform quantum circuits as $\Bpoly$. In this work, we ask: Does $\QCMA\subseteq \Bpoly$ imply $\QCPi_2=\QCSi_2$? Unfortunately, generalizing the proof of the Karp-Lipton theorem is problematic for the same ``$\exists\boldsymbol{\cdot}\BPP$ versus \MA\ phenomenon'' encountered in extending Toda's result. Namely, the proof of Karp-Lipton proceeds by fixing the outer, universally quantified, proof of a $\pit$ machine, and applying the $\class{NP} \subseteq \Ppoly$ hypothesis to the resulting \NP\ computation. However, for $\QCPi_2$, it is not clear that fixing the outer, universally quantified, proof yields a \QCMA\ computation; thus, it is not obvious how to use the hypothesis $\QCMA\subseteq \Bpoly$.

To sidestep this, our approach is to strengthen the hypothesis. Specifically, using the results of~\cite{JKNN12} on perfect completeness for \QCMA, fixing the outer proof of a $\QCPi_2$ computation can be seen to yield a \PrQCMA\ ``decision problem'', where by ``decision problem'', we mean no proofs for the \PrQCMA\ verifier are accepted within the promise gap. Here, \PrQCMA\ is \QCMA\ with {inverse} exponentially small promise gap. We hence obtain the following.

\begin{theorem}[A quantum-classical Karp-Lipton theorem]\label{thm:QCKL}
If $\PrQCMA \subseteq \Bpoly$, then {${\QCPi_2 = \QCSi_2}$}.
\end{theorem}

\noindent To give this result context, we also show that $\PrQCMA \subseteq \NP^{\PP}$ (Lemma~\ref{lem:bnd}). However, whether ${\QCPi_2 = \QCSi_2}$ collapses \QCPH\ remains open due to the same ``$\exists\boldsymbol{\cdot}\BPP$ versus \MA\ phenomenon''.

\subsection{Related work}
\label{sec:rel}

As far as we are aware, Yamakami~\cite{Y02} was the first to consider a quantum version of $\PH$. His version differs from our setting in that it considers quantum Turing machines (we use quantum circuits) and quantum {inputs} (we use classical inputs, like QMA). The next work, by Gharibian and Kempe~\cite{GK12}, introduced and studied ${\rm cq}\mhyphen\sigmatwo$, defined as our $\QCSi_2$ except with a quantum universally quantified proof. \cite{GK12} showed completeness and hardness of approximation results for ${\rm cq}\mhyphen\sigmatwo$ for (roughly) the following problem: What is the smallest number of terms required in a given local Hamiltonian for it to have a frustrated ground space? More recently, Lockhart and Gonz\'{a}lez-Guill\'{e}n~\cite{LG17} considered a hierarchy (denoted $\QCPH'$ here) which \emph{a priori} appears identical to our \QCPH, but is apparently not so due to the ``$\exists\boldsymbol{\cdot}\BPP$ versus \MA\ phenomenon'', which we now  discuss briefly (see also Remark~\ref{rem:langvsprom}).

In this work, the ``$\exists\boldsymbol{\cdot}\BPP$ versus \MA\ phenomenon'', refers to the following discrepancy: Unlike with \MA,  \emph{all} proofs in an $\exists\boldsymbol{\cdot}\BPP$ system {must} be accepted with probability at least $2/3$ or at most $1/3$ (i.e. no proof is accepted with probability in the gap $(1/3,2/3)$). The quantum analogue of this phenomenon yields the open question: Is $\exists\boldsymbol{\cdot}\BQP$ {(which equals $\NP^{\BQP}$)} equal to \QCMA? For this reason, it is not clear whether \QCPH\ equals $\QCPH'$, where {$\QCPH'$ is defined recursively as}
$\QCSi_1' = \exists\boldsymbol{\cdot}\BQP$,  $\QCPi_1' = \forall\boldsymbol{\cdot}\BQP$, and
\begin{equation*}
\forall i\geq 1, \QCSi_i' = \exists\boldsymbol{\cdot}\QCPi_{i-1}'; \qquad  \QCPi_i' = \forall\boldsymbol{\cdot}\QCSi_{i-1}'.
\end{equation*}
Thus, in our work $\QCSi_1=\QCMA$, but in~\cite{LG17} $\QCSi'_1=\exists\boldsymbol{\cdot}\BQP$. The advantage of the latter definition is that one avoids the recursion problems discussed earlier --- e.g., fixing the first existential proof in $\QCSi_2'$ \emph{does} reduce the problem to a {$\QCPi_1'$} computation, unlike the case with $\QCSi_2$. Hence, recursive arguments from the context of \PH\ can be extended to show that, for instance, $\QCPH'$ collapses to $\QCSi_2'$ when $\QCSi_2' = \QCPi_2'$. On the other hand, the advantage of our definition of \QCPH\ is that it generalizes the natural quantum complexity class \QCMA. 

Let us also remark on Toda's theorem in the context of $\QCPH'$ (for clarity, Toda's theorem is not studied in~\cite{LG17}). The recursive definition of $\QCPH'$ allows one to obtain Toda's $\Po^\pp$ upper bound for $\QCPH'$ with a simple argument:
\begin{equation*}
\forall i, \; \QCSi_i' = \NP^{\NP^{\iddots^{\BQP}}} = \sigmai^\BQP \quad\Longrightarrow\quad \forall i, \QCSi_i' \subseteq \left(\Po^{\pp}\right)^{\BQP} = \Po^{\pp},
\end{equation*}
{where the first equality holds due to the recursive definition of $\QCSi_i'$ (but is not known to hold for our $\QCSi_i$), the implication arises by relativizing Toda's theorem, and the last equality holds as \BQP\ is low for \pp~\cite{FR99}}. In contrast, our Theorem~\ref{thm:QCTODA} yields $\QCPH\subseteq\Po^{\pp^{\pp}}$, raising the question: is $\QCPH' = \QCPH$?  A positive answer may help shed light on whether $\exists\boldsymbol{\cdot}\BQP$ equals $\QCMA$; we leave this for future work.

Finally, a quantum version of the Karp-Lipton theorem was covered by Aaronson and Drucker in~\cite{AD14} and further improved by Aaronson, Cojocaru, Gheorghiu, and Kashefi~\cite{ACGK17}, where the consequences of $\NP$-complete problems being solved by small quantum circuits with polynomial sized quantum advice were considered.
Their results differ from ours in that different hierarchies are studied, and in their use of quantum advice as opposed to our use of classical advice. Other Karp-Lipton style results for \PH\ involving classes beyond \NP\ show a collapse of \PH\ to \MA\ (usually) if either $\pp$~\cite{LFKN92, Vin05}, $\Po^{\cp}$ or $\PSPACE$~\cite{KL80} has $\Ppoly$ circuits.

\subsection{Recent observations and open questions}\label{sscn:open}
\subsubsection{Recent observations}
Upon release of the current article, Sanketh Menda, Harumichi Nishimura, and John Watrous (whom we thank) made the observation that $\QSi_2=\QRGone$, where $\QRGone$ captures one-round zero-sum quantum games~\cite{QRGone}. Briefly, this equivalence follows immediately since $\QRGone$ from \cite{QRGone} can be defined as in Equation~(\ref{eq:alpha1}), but restricted to just the first two proofs, $\rho_1$ and $\rho_2$.\footnote{The use of \emph{mixed} state proofs in our definition of $\QSi_2$ and in Equation~\ref{eq:alpha1} is crucial for this equivalence.} This insight has led to some remarkable immediate corollaries regarding $\QSi_2$ and $\QPi_2$, which we now discuss.

The starting point for the discussion is that, as done for $\QRGone$ in~\cite{QRGone}, one can apply an extension of von Neumann's Min-Max Theorem~\cite{vN28} to conclude in Equation~(\ref{eq:alpha1}) that
 \[
    \max_{\rho_1}\min_{\rho_2}\inner{C}{\rho_1 \otimes \rho_2}=\min_{\rho_2}\max_{\rho_1}\inner{C}{\rho_1 \otimes \rho_2}.
 \]
In other words, $\QSi_2=\QPi_2$. In addition,~\cite{QRGone} shows $\QRGone\subseteq\PSPACE$. We thus immediately have the following.
\begin{corollary}\label{cor:wow}
    $\QSi_2=\QPi_2=\QRGone\subseteq\PSPACE$.
\end{corollary}

\noindent \emph{Relation to current work.} For the standard completeness-soundness gap regime ($s = 2/3, c = 1/3$), Corollary~\ref{cor:wow} improves upon our result of Theorem~\ref{thm:infEXP} (which recall showed $\QSi_2,\QPi_2\subseteq \EXP$). However, Theorem~\ref{thm:infEXP} and its proof, are still useful for the results in this paper: First, Theorem~\ref{thm:infEXP} works in the very small completeness-soundness gap regime. Second, the proof technique of Theorem~\ref{thm:infEXP} allows us to prove Theorem~\ref{IntroThm:NEXP} (e.g. $\QSi_3\subseteq \NEXP$), which also holds in the very small completeness-soundness gap regime.\\
\\

\noindent\emph{Further important implications of Corollary~\ref{cor:wow}.}
\begin{enumerate}
    \item  (Showing a ``collapse theorem'' for QPH will be ``hard'') One of the most frequently used results about $\PH$ is that if $\sigmatwo=\pit$, then $\PH$ collapses to $\sigmatwo$. Does a quantum analogue of this statement hold for \QPH? Corollary~\ref{cor:wow} yields the following.
        \begin{corollary}
            If $\QSi_2=\QPi_2$ implies $\QPH=\QSi_2$, then $\QMAt\subseteq \PSPACE$.
        \end{corollary}
        This follows immediately since recall $\QMAt\subseteq\QSi_3\subseteq\QPH$. Thus, proving such a ``collapse theorem'' for \QPH~ would require a breakthrough regarding the complexity characterization of $\QMAt$, which is believed to be challenging.

    \item (Separation between the second levels of PH and QPH) Corollary~\ref{cor:wow} also yields a separation of PH and QPH in the following sense, assuming the standard conjecture that PH is infinite.
        \begin{corollary}
            If $\QSi_2=\Sigma_2$ and $\QPi_2=\Pi_2$, then $\PH$ collapses to $\Sigma_2$.
        \end{corollary}
        This follows immediately since, as mentioned above, if $\sigmatwo=\pit$, then $\PH$ collapses to $\sigmatwo$. Thus, it is highly likely that either $\QSi_2\neq \Sigma_2$ or $\QPi_2\neq \Pi_2$ (or both).
\end{enumerate}

\subsubsection{Open questions}

{As far as general upper bounds on \QPH\ go, the currently best upper bound remains the naive one: The exponential-time analogue of PH, by which we mean constant-height towers of form $\NEXP^{\NEXP^{\NEXP\cdots}}$ (i.e. use each copy of NEXP to ``guess'' the next exponential size quantum proof, roughly speaking, just as in the proof that PH equals constant-height towers of NP oracles). An open question is to find a better upper bound on $\QPH$; we believe the naive bound to be loose.}

{One can also ask about the relationship between our $\QCPH$ and $\QPH$ classes and constant-height towers of the form $\QCMA^{\QCMA^{\QCMA\cdots}}$ (a ``QCMA-hierarchy'') and $\QMA^{\QMA^{\QMA\cdots}}$ (a ``QMA-hierarchy''), respectively. In this work, we have not studied the QCMA- and QMA-hierarchies, as they involve \emph{quantum} machines making oracle queries, and this in itself would presumably need to be correctly defined. For example, are superposition queries to the oracles allowed? Are the queries in-place? We believe this is an interesting avenue for future work. What we do observe here is that, since $\QMA \subseteq \PP$~\cite{KW00,Vy03,MW05,GY16}, if an appropriately defined QMA-hierarchy equals QPH, then it would presumably put QMA(2) in the Counting Hierarchy (CH) (and hence in PSPACE), which would again require a breakthrough in our understanding of unentangled quantum proofs (i.e. QMA(2)).}

{Finally, determining where in the complexity zoo $\QMA(2)$ belongs remains an important open question. Assuming alternating quantifiers \emph{do} add proving power to quantum proofs (the analogous assumption for classical proofs is widely believed), our work shows $\QMA(2)$ is strictly contained in NEXP. Can this statement be strengthened?}

\paragraph{Organization:}
We begin in Section~\ref{sec:preliminaries} by formally introducing relevant complexity classes. In Section~\ref{sec:toda} we show a quantum-classical analogue of Toda's theorem. Section~\ref{scn:fullyquantum} gives upper bounds on levels of \QPH, and Section~\ref{sec:karp-lipton} shows a Karp-Lipton-type theorem.

\section{{Definitions, preliminaries, and basic properties}}
\label{sec:preliminaries}

We begin by recalling the definition of uniformly-generated families of quantum circuits.

\begin{definition}[Polynomial-time uniform family of quantum circuits]
	A family of quantum circuits $\{ V_n \}_{n \in \natural}$ is said to be uniformly generated in polynomial time if there exists a polynomially bounded function $t:\natural\mapsto\natural$ and a deterministic Turing machine $M$ acting as follows. For every $n$-bit input $x$, $M$ outputs in time $t(n)$ a description of a quantum circuit $V_n$ (consisting of $1$-and $2$-qubit gates) that takes the all-zeros state as ancilla and outputs a single qubit. We say $V_n$ \emph{accepts} if measuring its output qubit in the computational basis yields $1$.
\end{definition}

Throughout this paper, we study \emph{promise problems}. A promise problem is a pair $A=(\ayes,\ano)$ such that $\ayes, \ano\subseteq\set{0,1}^\ast$, $\ayes \cup \ano \subset \set{0, 1}^*$ and $\ayes \cap \ano = \emptyset$. We now formally define each level of our quantum-classical polynomial hierarchy below.

\begin{definition}[$\QCSi_i$]\label{def:QCSigmam}
	Let $A=(\ayes,\ano)$ be a promise problem. We say that $A$ is in $\QCSi_i(c, s)$ for polynomial-time computable functions $c, s: \natural \mapsto [0, 1]$ if there exists a polynomially bounded function $p:\natural\mapsto\natural$ and a polynomial-time uniform family of quantum circuits $\{V_n\}_{n \in \natural}$ such that for every $n$-bit input $x$, $V_n$ takes in classical proofs ${y_1}\in \set{0,1}^{p(n)}, \ldots, {y_i}\in \set{0,1}^{p(n)}$ and outputs a single qubit, {such that:}
	\begin{itemize}
		\item Completeness: $x\in \ayes$ $\Rightarrow$ $\exists y_1 \forall y_2 \ldots Q_i y_i$ s.t. $\operatorname{Prob}[V_n \text{ accepts } (y_1, \ldots, y_i)] \geq c$.
		\item Soundness: $x\in \ano$ $\Rightarrow$ $\forall y_1 \exists y_2 \ldots \overline{Q}_i y_i$ s.t. $\operatorname{Prob}[V_n \text{ accepts } (y_1, \ldots, y_i)] \leq s$.
	\end{itemize}
	Here, $Q_i$ equals $\exists$ when $m$ is odd and equals $\forall$ otherwise and $\overline{Q}_i$ is the complementary quantifier to $Q_i$.
	\begin{equation}
	\text{Define } \; \QCSi_i := \bigcup_{{ c - s \in \Omega(1/\poly(n))}} \QCSi_i(c, s).
	\end{equation}
\end{definition}

Notice that the first level of this hierarchy corresponds to $\QCMA$. The complement of the $i^{th}$ level of the hierarchy, $\QCSi_i$, is the class $\QCPi_i$ defined below.

\begin{definition}[$\QCPi_i$]\label{def:QCPim}
	Let $A=(\ayes,\ano)$ be a promise problem. We say that $A \in \QCPi_i(c, s)$ for polynomial-time computable functions $c, s: \natural \mapsto [0, 1]$ if there exists a polynomially bounded function $p:\natural\mapsto\natural$ and a polynomial-time uniform family of quantum circuits $\{V_n\}_{n \in \natural}$ such that for every $n$-bit input $x$, $V_n$ takes in classical proofs ${y_1}\in \set{0,1}^{p(n)}, \ldots, {y_i}\in \set{0,1}^{p(n)}$ and outputs a single qubit, {such that:}
	\begin{itemize}
		\item Completeness: $x\in \ayes$ $\Rightarrow$ $\forall y_1 \exists y_2 \ldots Q_i y_i$ s.t. $\operatorname{Prob}[V_n \text{ accepts } (y_1, \ldots, y_i)] \geq c$.
		\item Soundness: $x\in \ano$ $\Rightarrow$ $\exists y_1 \forall y_2 \ldots \overline{Q}_i y_i$ s.t. s.t. $\operatorname{Prob}[V_n \text{ accepts } (y_1, \ldots, y_i)] \leq s$.
	\end{itemize}
	Here, $Q_i$ equals $\forall$ when $m$ is odd and equals $\exists$ otherwise and $\overline{Q}_i$ is the complementary quantifier to $Q_i$.
	\begin{equation}
	\text{Define } \; \QCPi_i := \bigcup_{{c - s \in \Omega(1/\poly(n))}} \QCPi_i(c, s).
	\end{equation}
\end{definition}

Now the corresponding quantum-classical polynomial hierarchy is defined as below.

\begin{definition}[Quantum-classical polynomial hierarchy]\label{def:QCPH}
	\[ \QCPH = \bigcup_{m \in \mathbb{N}} \; \QCSi_i = \bigcup_{m \in \mathbb{N}} \; \QCPi_i. \]
\end{definition}

\noindent A few remarks are in order. First, by encoding a polynomial time predicate into a quantum verification circuit, one can see that (where $\sigmai$ and $\pii$ refer to the $i^{th}$ level of the corresponding classical polynomial hierarchy)
\begin{align*}
\forall i, \; \sigmai \subseteq \QCSi_i, \quad \pii \subseteq \QCPi_i \quad \text{ and } \quad \PH \subseteq \QCPH.
\end{align*}

\noindent Second, a natural question is to what extent the completeness and soundness parameters of $\QCSi_i$ and $\QCPi_i$ can be improved. Towards achieving one-sided error, we apply known techniques to prove that ``every other level'' (see Theorem~\ref{res:err} for a formal statement) has \emph{perfect completeness} (i.e. we can improve the completeness parameter to $c = 1$), in addition to every level having inverse exponentially small soundness. This is shown using techniques from the proof of the following theorem.

\begin{theorem}[Jordan, Kobayashi, Nagaj, Nishimura~\cite{JKNN12}] \label{QCMAPC}
	$\class{QCMA}$ has perfect completeness i.e.
	\begin{equation}
		\QCMA = \QCMA(1, 1 - 1/\poly(n)).
	\end{equation}
\end{theorem}

\noindent The proof of the above result {starts} by choosing a suitable gate-set for the $\QCMA$ verifier, i.e., Hadamard, Toffoli and CNOT gates~\cite{S03, A03}. This ensures that the acceptance probability for any proof $y$ can be expressed as $k/2^{\ell(|x|)}$ for an integer $k \in \{0, \ldots, 2^{\ell(|x|)} \}$ and a polynomially bounded integer function $\ell(|x|)$.
The verifier then asks the prover to send $k$ (expressed as a polynomial-size bit string) along with the classical proof. When $k$ is above a certain threshold, the verifier  chooses one of two tests with equal {probability}: (a) run the original verification circuit or (b) trivially accept with probability $> k/2^{\ell(|x|)}$.
This allows for the completeness to be reduced to exactly $1/2$ while the soundness is strictly bounded below $1/2$. Then by using the quantum rewinding technique~\cite{W09_3}, {$c$ can be boosted} to exactly $1$. The {ideas in this proof have been} adapted to several similar scenarios (see, e.g., \cite{KlGN12, GKS16}). We state our result below.

\begin{theorem}\label{res:err}
	For polynomially bounded functions $r, q: \natural \mapsto \natural$ and polynomial-time computable functions $c, s: \natural \mapsto [0, 1]$ such that for any $n$-bit input $c(n) - s(n) \geq 1/q(n)$, we have:
	\[ \begin{array}{rlllll}
	\textup{ For }  i \textup{ even:} &\QCSi_i(c, s) & = & \QCSi_i(1- 2^{-r}, 2^{-r}),  \\
	&\QCPi_i(c, s) & = & \QCPi_i(1, 2^{-r}), \\
	\textup{ for }  i \textup{ odd:} &\QCSi_i(c, s) & = & \QCSi_i(1, 2^{-r}), \\
	&\QCPi_i(c, s) & = & \QCPi_i(1 - 2^{-r}, 2^{-r}).
	\end{array} \]
\end{theorem}

\begin{proof}[Proof Sketch.]
To achieve perfect completeness (i.e. $c = 1$), the idea is to append to the register of the last proof (which must be an existential quantifier for this to work) a classical register containing the acceptance probability of the verification circuit $C$. {Specifically, for level $i$, for any set of $i-1$ proofs $y_1, \ldots, y_{i-1}$, the final (existential) proof $y_i$ is augmented with $k$, such that $\Pr[C({y_1, \ldots, y_i}) = 1] = k/2^{\ell(|x|)}$ (that this probability is rational is due solely to the use of an appropriate universal gate set, as done for Theorem~\ref{QCMAPC}, and is independent of how each $y_i$ for $i\in\set{1,\ldots, i-1}$ is quantified).} Then the proof of Theorem~\ref{QCMAPC} in~\cite{JKNN12} proves the result.
	The error reduction follows from standard arguments.
\end{proof}

\noindent Notice that by explicitly emulating the technique from \cite{JKNN12} we are using it as a \emph{white-box} and not a black box reduction. Hence, the issues discussed in Section~\ref{sec:intro} the arise from ``fixing proofs'' does not apply here.  We leave as an open problem the question of obtaining perfect completeness for the remaining levels of the hierarchy. This seems like a considerably harder
problem, with current proof techniques requiring the last quantifier to be existential.

Now, we move on to defining the fully quantum hierarchy.

\begin{definition}[$\QSi_i$]\label{def:QSigmam}
	A promise problem $A=(\ayes,\ano)$ is in $\QSi_i(c, s)$ for polynomial-time computable functions $c, s: \natural \mapsto [0, 1]$ if there exists a polynomially bounded function $p:\natural\mapsto\natural$ and a polynomial-time uniform family of quantum circuits $\{V_n\}_{n \in \natural}$ such that for every $n$-bit input $x$, $V_n$ takes $p(n)$-qubit density operators $\rho_1, \ldots, \rho_i$ as quantum proofs and outputs a single {qubit, then:}
	\begin{itemize}
		\item Completeness: If $x\in\ayes$, then $\exists \rho_1 \forall \rho_2 \ldots Q_i \rho_i$ such that $V_n$ accepts $({\rho_1 \otimes \rho_2 \otimes \cdots \otimes  \rho_i})$ with probability $\geq c$.
		\item Soundness: If $x\in\ano$, then $\forall \rho_1 \exists \rho_2 \ldots \overline{Q}_i \rho_i$ such that $V_n$ accepts $({\rho_1 \otimes \rho_2 \otimes \cdots \otimes  \rho_i})$ with probability $\leq s$.
	\end{itemize}
	Here, $Q_i$ equals $\forall$ when $m$ is even and equals $\exists$ otherwise, and $\overline{Q}_i$ is the complementary quantifier to $Q_i$.
	\begin{equation}
	\text{Define } \; \QSi_i = \bigcup_{\substack{c-s \in \Omega(1/\poly(n))}} \QSi_i(c, s).
	\end{equation}
\end{definition}

\noindent A few comments are in order: (1) {In contrast to the standard quantum circuit model, here we allow \emph{mixed} states as inputs to $V_n$; this can be formally modelled via the mixed state framework of~\cite{AKN98}.} (2) Clearly, $\QSi_1 = \QMA$. (3) We recover the definition of $\QMA(k)$ by ignoring the $\rho_i$ proofs, for $i$ even, in the definition of $\class{Q$\Sigma$}_{2k}$.

\begin{definition}[$\QPi_i$]\label{def:QPim}
	A promise problem $A=(\ayes,\ano)$ is in $\QPi_i(c, s)$ for polynomial-time computable functions $c, s: \natural \mapsto [0, 1]$ if there exists a polynomially bounded function $p:\natural\mapsto\natural$ and a polynomial-time uniform family of quantum circuits $\{V_n\}_{n \in \natural}$ such that for every $n$-bit input $x$, $V_n$ takes $p(n)$-qubit density operators $\rho_1, \ldots, \rho_i$ as quantum proofs and outputs a single {qubit, then:}
	\begin{itemize}
		\item Completeness: If $x\in\ayes$, then $\forall \rho_1 \exists \rho_2 \ldots Q_i \rho_i$ such that $V_n$ accepts $({\rho_1 \otimes \rho_2 \otimes \cdots \otimes  \rho_i})$ with probability $ \geq c$.
		\item Soundness: If $x\in\ano$, then $\exists \rho_1 \forall \rho_2 \ldots \overline{Q}_i \rho_i$ such that $V_n$ accepts $({\rho_1 \otimes \rho_2 \otimes \cdots \otimes  \rho_i})$ with probability $\leq s$.
	\end{itemize}
	Here, $Q_i$ equals $\exists$ when $m$ is even and equals $\forall$ otherwise, and $\overline{Q}_i$ is the complementary quantifier to $Q_i$.
	\begin{equation}
	\text{Define } \; \QPi_i = \bigcup_{c-s \in \Omega(1/\poly(n))} \QPi_i(c, s).
	\end{equation}
\end{definition}

The fully quantum polynomial hierarchy can now be defined as follows.

\begin{definition}[Quantum Polynomial-Hierarchy]\label{def:QPH}
	\[ \QPH = \bigcup_{m \in \mathbb{N}} \; \QSi_i = \bigcup_{m \in \mathbb{N}} \; \QPi_i. \]
\end{definition}

Finally, we recall the definition of $\pp$ as it repeatedly used up in various results throughout this work.

\begin{definition}[$\class{PP}$]
	\label{def:PP1}
	A language $\calL$ is in $\class{PP}$ if there exists a probabilistic polynomial time Turing machine $M$ such that  $x \in \calL \iff \Pr[M(x) = 1] > 1/2$.
\end{definition}

\section{A quantum-classical analogue of Toda's theorem}
\label{sec:toda}
In this section, we show an analogue of Toda's theorem to bound the power of $\QCPH$ (Theorem~\ref{thm:QCTODA}, Section~\ref{sscn:bounding}), and give evidence that the bound of Theorem~\ref{thm:QCTODA} is likely the best possible using our specific proof approach (Section~\ref{sscn:ppcomp}, Theorem~\ref{thm:ppcomplete}).

\subsection{\PrBQP}
\label{sec:prbqp}

Our proof of a ``quantum-classical Toda's theorem'' requires us to define the $\PrBQP$ class, which we do now.
\begin{definition}[$\PrBQP(c,s)$]\label{def:PBQP}
    A promise problem $A=(\ayes,\ano)$ is contained in $\PrBQP(c,s)$ for polynomial-time computable functions $c,s:\natural\mapsto[0,1]$ if there exists a polynomially bounded function $p:\natural\mapsto\natural$ such that $\forall \ell \in \natural, \ c(\ell) - s(\ell) \geq 2^{-p(\ell)}$, and a polynomial-time uniform family of quantum circuits $\{V_n\}_{n \in \natural}$ whose input is the all zeroes state and output is a single qubit. Furthermore, for an $n$-bit input $x$:
    \begin{itemize}
     \item Completeness: If $x \in \ayes$, then $V_n$ accepts with probability at least $c(n)$.
     \item Soundness: If $x \in \ano$, then ${V_n}$ accepts with probability at most $s(n)$.
    \end{itemize}
\end{definition}

\noindent In contrast, \BQP\ is defined such that the completeness and soundness parameters are $2/3$ and $1/3$, respectively (alternatively, the gap is least an inverse polynomial in $n$). We now give a useful observation and lemma.

\begin{obs}[Rational acceptance probabilities]\label{obs:gate}
    By fixing an appropriate universal gate set (e.g. Hadamard and Toffoli~\cite{A03}) for the description of $V_n$ in Definition~\ref{def:PBQP}, we assume henceforth, without loss of generality, that the acceptance probability of $V_n$ is a rational number that can be represented using at most $\poly(n)$ bits (this observation was used in the proof that $\QCMA$ has perfect completeness i.e., $c=1$~\cite{JKNN12}, {stated as Theorem~\ref{QCMAPC} here}).
\end{obs}

The following lemmas {help to characterize the complexity of $\PrBQP$}.

\begin{lemma}\label{l:containedPP}
    For any polynomial $p$, if $c - s \geq 1/2^{p(n)}$, then
    \begin{equation}
    	\PrBQP(c,s) \subseteq \pp
    \end{equation}
    when $c$ and $s$ are computable in polynomial time in the size of the input, $n$.
\end{lemma}

\begin{proof}[Proof sketch.]
	Recall that the complexity class \class{PQP} is defined as \PP\ except with a uniform quantum circuit family $\set{Q_n}$ in place of a probabilistic Turing machine, i.e. for \textsf{YES} (\textsf{NO}) instances $Q_n$ accepts with probability $>1/2$ ($\leq 1/2$). Consider any $\PrBQP(c,s)$ circuit $V_x$ as in Definition~\ref{def:PBQP}. Then, by flipping a coin with appropriately chosen bias $\gamma\in \mathbb{Q}$ and choosing to either accept/reject with probability $\gamma$ and run $Q_n$ with probability $1-\gamma$, one may map $c,s$ to polynomial-time computable functions $c',s'$ such that
	\begin{equation}
	c'>1/2, \; s' \leq 1/2, \; \text{ and } \; c'-s'\in\Theta(c-s)
	\end{equation}
	(roughly, one loses about a factor of at most approximately $1/2$ in the gap). Thus,
	\begin{equation}
	\PrBQP(c,s)\subseteq \class{PQP}=\class{PP},
	\end{equation}
	where the last equality is shown in~\cite{W09_2}.
\end{proof}

\noindent As an aside, we remark the following.
\begin{corollary}\label{cor:equalsPP}
    Let $\mathbb{P}$ denote the set of all polynomials $p:\mathbb{N}\mapsto \mathbb{N}$. Then,
    \[
        \bigcup_{p\in \mathbb{P}}\PrBQP\left(\frac{1}{2}+\frac{1}{2^{p(n)}},\frac{1}{2}\right)=\PP.
    \]
\end{corollary}

To prove Corollary~\ref{cor:equalsPP}, we need the classical counterpart to $\PrBQP$, denoted \class{Precise-BPP}. Accordingly, we define \class{Precise-BPP} analogous to Definition~\ref{def:PBQP} except by replacing the quantum circuit family $\set{V_n}_{n \in \natural}$ with a deterministic polynomial-time Turing machine which takes in polynomially many bits of randomness.

\begin{proof}
	The direction $\subseteq$ is given by Lemma~\ref{l:containedPP}. For the reverse containment, note that
	\begin{equation}
	\PP=\bigcup_{p}\class{PreciseBPP}\left(\frac{1}{2}+\frac{1}{2^{p(n)}},\frac{1}{2}\right),
	\end{equation}
	since PP can be defined as the set of decision problems of the form: Given as input a polynomial-time non-deterministic Turing machine $N$, do more than half of $N$'s computational paths accept? The claim now follows, since for all $c,s$ as in Definition~\ref{def:PBQP}, clearly $\class{PreciseBPP}(c,s)\subseteq \PrBQP(c,s)$.
\end{proof}

\noindent Note that this proof does not go through as is (assuming $\PP \neq \class{co-NP}$) when we fix (say) completeness $c=1$ and soundness $s=1-2^{-p(n)}$, for some polynomial $p$. This is because
\begin{equation}
	\bigcup_{p}\class{PreciseBPP}(1,1-\frac{1}{2^{p(n)}})=\class{co-NP}.
\end{equation}
Similarly, setting $s=0$ and $c=2^{-p(n)}$ yields $\NP$.

Finally, we define the promise problem $\BQPC(c, s)$, which is trivially $\PrBQP(c,s)$-complete when $c-s$ is an inverse exponential. 

\begin{definition}[$\BQPC(c,s)$]\label{def:PBQPCIRCUIT} Parameters $c,s:\natural\mapsto[0,1]$ are polynomial-time computable functions such that $c>s$.
	\begin{itemize}
		\item (Input) A classical description of quantum circuit $V_n$ (acting on $n$ qubits, consisting of $\poly(n)$ $1$ and $2$-qubit gates), taking in the all-zeroes state, and outputting a single qubit.
		\item (Output) Decide if $\operatorname{Pr}[V_n\text{ accepts }]\geq c$ or $\leq s$, assuming one of the two is the case.
	\end{itemize}
\end{definition}

\subsection{Bounding the power of \texorpdfstring{$\QCPH$}{the quantum-classical hierarchy}}\label{sscn:bounding}
Classically, PH can be defined in terms of the existential ($\exists$) and universal ($\forall$) operators; while it is not clear that one can also define \QCPH\ using these operators, they nevertheless prove useful in bounding the power of \QCPH.

\begin{definition}[Existential and universal quantifiers~\cite{W76, AW93}]\label{def:exists}
     For $\mathcal{C}$ a class of languages, $\exists\boldsymbol{\cdot}\mathcal{C}$ is defined as the set of languages $L$ such that there is a polynomial $p$ and set $A\in \mathcal{C}$ such that for input $x$,
     $
        x\in L \Leftrightarrow \left[\exists y\, (\abs{y}\leq p(\abs{x}))\text{ and }\langle x,y\rangle \in A\right].
     $
     The set $\forall\boldsymbol{\cdot}\mathcal{C}$ is defined similarly with $\exists$ replaced with $\forall$.
\end{definition}
\begin{remark}[{Languages versus promise problems}] \label{rem:langvsprom}
    Directly extending Definition~\ref{def:exists} to promise problems, gives rise to subtle issues. To demonstrate, recall that $\exists\boldsymbol {\cdot} \Po=\NP$. Then, let $(L,A)$ for $L\in\exists\boldsymbol\cdot\Po=\NP$ and $A\in\Po$ be as in Definition~\ref{def:exists}, such that $T_A$ is a polynomial-time Turing machine deciding $A$. If $x\in L$, there exists a bounded length witness $y^*$ such that $T_A$ accepts $\langle x,y^*\rangle$ and, for all $y'\neq y^*$, $T_A$ by definition either accepts or rejects $\langle x,y'\rangle$. Now consider instead $\exists\boldsymbol\cdot \BPP$, which \emph{a priori} seems equal to Merlin-Arthur (\MA). Applying the same definition of $\exists$, we should obtain a \class{BPP}  machine $T_A$ such that if $x\in L$, then for all $y'\neq y^*$, $T_A$ \emph{either accepts or rejects} $\langle x,y'\rangle$. But this means, by definition of \class{BPP}, that $\langle x,y'\rangle$ is either accepted or rejected with probability at least $2/3$, respectively. (Equivalently, for any fixed $y$, the machine $T_{A,y}$ must be a \class{BPP} machine.) Unfortunately, the definition of \MA\ \emph{makes no such promise} --- any $y'\neq y^*$ can be accepted with arbitrary probability when $x$ is a \textsf{YES} instance. Indeed, whether $\exists \boldsymbol \cdot \BPP = \MA$ remains an open question~\cite{FFKL03}.
\end{remark}

The following lemma is the main contribution of this section. To set context, adapting the ideas from Toda's proof of $\PH\subseteq\Po^{\PP}$ to \QCPH\ is problematic for at least two reasons:
\begin{enumerate}
    \item Remark~\ref{rem:langvsprom} says that it is not necessarily true that by fixing a proof $y$ to an \MA\ (resp. \QCMA) machine, the resulting machine is a \BPP\ (resp. \BQP) machine. This prevents the direct extension of recursive arguments, say from~\cite{T91}, to this regime.
    \item The ``Quantum Valiant Vazirani (QVV)'' theorem for QCMA (and MA)~\cite{ABBS08} is not a many-one reduction, but a \emph{Turing} reduction.
         Specifically, it produces a set of quantum circuits $\set{Q_i}$, \emph{at least one of which} is guaranteed to be a \textsf{YES} instance of some Unique-QCMA promise problem $\Gamma$ if the input $\Pi$ to the reduction was a \textsf{YES} instance. Unfortunately, some of the $Q_i$ may violate the promise gap of $\Gamma$, which implies that when such $Q_i$ are {substituted} into the Unique-QCMA oracle O, O returns an arbitrary answer. This does not pose a problem in~\cite{ABBS08}, as one-sided error suffices for that reduction --- so long as O accepts at least one $Q_i$, one safely concludes $\Pi$ was a \textsf{YES} instance. In the setting of Toda's theorem, however, the use of \emph{alternating} quantifiers turns this one-sided error into two-sided error. This renders the output of O useless, as one can no longer determine whether $\Pi$ was a \textsf{YES} or \textsf{NO} instance.
\end{enumerate}
To sidestep these issues, we adapt a high-level idea from~\cite{GY16}: With the help of an appropriate \emph{oracle}, one can sometimes detect ``invalid proofs'' (i.e. proofs in promise gaps of bounded error verifiers) and ``remove'' them. Indeed, we show that using a \PP\ oracle, one can eliminate the promise-gap of $\QCPH$ altogether, thus overcoming the limitations given above. This is accomplished by the following ``Cleaning Lemma''. We also show subsequently that it is highly unlikely for an oracle weaker than \PP\ to suffice for our particular proof technique (see Remark~\ref{rem:stronger} and Section~\ref{sscn:ppcomp}).

\begin{lemma}[Cleaning Lemma]\label{l:clean}
    For all $i\geq 0$,
    \[
        \QCSi_i\subseteq \exists\boldsymbol{\cdot}\forall\boldsymbol{\cdot}\cdots\boldsymbol{\cdot} Q_i\boldsymbol{\cdot} \class{P}^{\PrBQP}\subseteq \exists\boldsymbol{\cdot}\forall\boldsymbol{\cdot}\cdots\boldsymbol{\cdot} Q_i\boldsymbol{\cdot} \class{P}^{\PP},
    \]
    where $Q_i=\exists$ ($Q_i=\forall$) if $i$ is odd (even). An analogous statement holds for $\QCPi_i$.
\end{lemma}

\begin{proof}
	Let $C$ be a $\QCSi_i$ verification circuit for a promise problem $\Pi$. Let $C_{y^*_1,\ldots,y^*_i}$ denote the quantum circuit obtained from $C$ by fixing values $y_1^*,\ldots,y_i^*$ of the $i$ classical proofs. In general, nothing can be said about the acceptance probability $p_{y^*_1,\ldots,y^*_i}$ of $C_{y^*_1,\ldots,y^*_i}$, except that, by Observation~\ref{obs:gate}, $p_{y^*_1,\ldots,y^*_i}$ is a rational number representable using $p(n)$ bits for some fixed polynomial $p$. Let $S$ denote the set of all rational numbers in $[0,1]$ representable using $p(n)$ bits of precision. (Note $\abs{S}\in\Theta(2^{p(n)})$.) Then, for any $a,b\in S$ with $a>b$, the triple $(C_{y^*_1,\ldots,y^*_i},a,b)$ is a valid $\BQPC(a,b)$ instance, in that $C_{y^*_1,\ldots,y^*_i}$ accepts with probability at least $a$ or at most $b$ for $a-b$ an inverse exponential. It follows that using binary search (by varying the values $a,b\in S$ with $a>b$) in conjunction with $\poly(n)$ calls to a $\BQPC(a,b)$ oracle, we may exactly and deterministically compute $p_{y^*_1,\ldots,y^*_i}$. Moreover, since for all such $a> b$, $\BQPC(a,b)\in \PrBQP(a,b)$, Lemma~\ref{l:containedPP} implies a $\BQPC(a,b)$ oracle call can be simulated with a \PP\ oracle. Denote the binary search subroutine using the \PP\ oracle as $B$.
	
	Using $C$ and $B$, we now construct an oracle Turing machine $C'$ as follows. Given any proofs  $y_1^*,\ldots, y_i^*$ as input, $C'$ uses $B$ to compute $p_{y^*_1,\ldots,y^*_i}$ for $C_{y^*_1,\ldots,y^*_i}$. If $p_{y^*_1,\ldots,y^*_i}\geq c$, $C'$ accepts with certainty, and if $p_{y^*_1,\ldots,y^*_i}<c$, $C'$ rejects with certainty. Suppose that the circuits $C$ and $C'$ return $1$ when the accept and $0$ when they reject. Two observations: (1) Since by construction, for any fixed $y^*_1,\ldots,y^*_i$, $B$ makes only makes ``valid'' $\BQPC(a,b)$ queries (i.e. satisfying the promise of $\BQPC(a,b)$), $C'$ is a $\Po^{\PP}$ machine (cf. Observation~\ref{obs:not}). (2) Since $C'_{y^*_1,\ldots,y^*_i}$ accepts if $C_{y^*_1,\ldots,y^*_i}$ accepts with probability at least $c$, and since $C'_{y^*_1,\ldots,y^*_i}$ rejects if $C_{y^*_1,\ldots,y^*_i}$ accepts with probability at most $s$, we conclude that
	\begin{eqnarray}
	\exists y_1 \forall y_2 \cdots Q_i y_i \operatorname{Prob}[C(y_1,\ldots, y_i) = 1]\geq c &\Leftrightarrow& \exists y_1 \forall y_2 \cdots Q_i y_i \ C'(y_1,\ldots, y_i)=1  \label{eq:1} \\
	\forall y_1 \exists y_2 \cdots \overline{Q}_i y_i \operatorname{Prob}[C(y_1,\ldots, y_i) = 1]\leq s &\Leftrightarrow& \forall y_1 \exists y_2 \cdots \overline{Q}_i y_i \ C'(y_1,\ldots, y_i) = 0 \label{eq:2}.
	\end{eqnarray}
	\eqref{eq:1} and~\eqref{eq:2} imply that we can simulate $\Pi$ with a $\exists\boldsymbol{\cdot}\forall\boldsymbol{\cdot}\cdots\boldsymbol{\cdot} Q_i\boldsymbol{\cdot} \class{P}^{\PP}$ computation. The proof for $\QCPi_i$ is analogous.
\end{proof}

\begin{remark}[Possibility of a stronger containment]\label{rem:stronger}
	A key question is whether one may replace the $\PrBQP$ oracle in the proof of Lemma~\ref{l:clean} with a weaker \BQP\ oracle. For example, consider the following alternative definition for oracle Turing machine $C'$: Given proofs $y_1^*,\ldots, y_i^*$, $C'$ plugs $C_{y_1^*,\ldots, y_i^*}$ into a BQP oracle and returns the oracle's answers. It is easy to see that in this case, Equations~(\ref{eq:1}) and (\ref{eq:2}) hold. However, $C'$ is \emph{not} necessarily a $\Po^{\BQP}$ machine, since for some settings of $y_1^*,\ldots, y_i^*$, its input to the \BQP\ oracle may violate the BQP promise, hence making the output of $C'$ ill-defined. To further illustrate this subtle point, consider Observation~\ref{obs:not}. Moreover, in Section~\ref{sscn:ppcomp} we show that the task the $\PrBQP$ oracle is used for in Lemma~\ref{l:clean} is in fact PP-complete; thus, it is highly unlikely that one can substitute a weaker oracle into the proof above.
\end{remark}

\begin{obs}[When a \Po\ machine querying a \BQP\ oracle is not a $\Po^{\BQP}$ machine]\label{obs:not} The proof of the Cleaning Lemma uses a $\Po^{\PrBQP}$ machine. Let us highlight a subtle reason why using a weaker $\BQP$ oracle instead might be difficult (indeed, in Section~\ref{sscn:ppcomp} we show that the task we use the $\PrBQP$ oracle for is $\PP$-complete).   Let $M$ denote the trivially \BQP-complete problem of determining whether a given polynomial-sized quantum circuit $Q$ accepts with probability at least $2/3$, or accepts with probability at most $1/3$, with the promise that one of the two is the case. Now consider the following polynomial time computation, $\Pi$, which is given access to an oracle $O_M$ for $M$: $\Pi$ inputs the Hadamard gate $H$ into $O_M$ and outputs $O_M$'s answer. Does it hold that $\Pi\in \Po^{\BQP}$? No. Since $H$ violates the promise of BQP, i.e. measuring the output of $H$ yields $0$ or $1$ with equal probability, the oracle $O_M$ can answer $0$ or $1$ arbitrarily, and so the output of $\Pi$ is not well-defined. Having a well-defined output, however, is required for a $\Po^{O_K}$ computation, where $K$ is any promise class~\cite{G06}.
\end{obs}

Using standard techniques, we next show the following.

\begin{lemma}\label{l:contain}
    For all $i\geq 0$, the following holds true:
    \begin{eqnarray*}
    	\exists\boldsymbol{\cdot}\forall\boldsymbol{\cdot}\cdots\boldsymbol{\cdot} Q_i\boldsymbol{\cdot} \class{P}^{\PP}
    	&\subseteq&\sigmai^{\PP},\\
    	\forall\boldsymbol{\cdot}\exists\boldsymbol{\cdot}\cdots\boldsymbol{\cdot} Q_i\boldsymbol{\cdot} \class{P}^{\PP}
    	&\subseteq&\pii^{\PP},
    \end{eqnarray*}
    where $Q_i=\exists$ (resp. $Q_i=\forall$) when $i$ is odd (resp. even) in the first containment and vice versa for the second containment.
\end{lemma}

\begin{proof}
	We show the first statement with containment in $\sigmai^{\pp}$; the second containment follows using an analogous proof. Let $\NP_i$ be defined recursively as $\NP_i := \NP^{\NP_{i-1}}$ with $\NP_1 := \NP$. We show that $\exists\boldsymbol{\cdot}\forall\boldsymbol{\cdot}\cdots\boldsymbol{\cdot} Q_i\boldsymbol{\cdot} \class{P}^{\PP} \subseteq \NP_i^{\PP}$, and then use the fact that {$\NP_i = \sigmai$ using the oracular definition for $\sigmai$}. Recall that MAJSAT is a \PP-complete language, where, given a Boolean formula $\phi$, one must decide if more than half of the possible assignments $x$ satisfy $\phi(x)=1$. For brevity, let $A_i$ denote the (trivially) $\exists\boldsymbol{\cdot}\forall\boldsymbol{\cdot}\cdots\boldsymbol{\cdot} Q_i\boldsymbol{\cdot} \class{P}^{\PP}$-complete language (under polynomial-time many-one reductions) --- given as input a polynomial time oracle Turing machine $T$ with access to a MAJSAT oracle, decide which of the two is the case:
\begin{eqnarray*}
		&&\exists y_1\forall y_2\cdots Q_i y_i \;\; \text{ such that }T\text{ accepts }\langle y_1,\ldots, y_i\rangle,\\
		&&\forall y_1\exists y_2\cdots \overline{Q}_i y_i\;\; \text{ such that }T\text{ rejects }\langle y_1,\ldots, y_i\rangle.
	\end{eqnarray*}
	Let $B_i$ denote the analogous trivially complete problem for $\forall\boldsymbol{\cdot}\exists\boldsymbol{\cdot}\cdots\boldsymbol{\cdot} Q_i\boldsymbol{\cdot} \class{P}^{\PP}$. We proceed by induction. The base case $i=0$ holds trivially since $\Sigma_0=\Po$ by definition. For the inductive step $i\geq 1$, let $L$ be a language in $\exists\boldsymbol{\cdot}\forall\boldsymbol{\cdot}\cdots\boldsymbol{\cdot} Q_{i}\boldsymbol{\cdot} \class{P}^{\PP}$. Then there exists a polynomial-time oracle Turing machine $T$, with access to a MAJSAT oracle, and such that $x\in L$ if and only if $\exists y_1$ such that
	\[
	\forall y_2 \exists y_3\cdots Q_i y_i \;\; \text{ such that } T \text{ accepts }\langle x,y_1,y_2,\cdots,y_i\rangle.
	\]
	By non-deterministically guessing $y_1$, it follows that $L\in \NP^{B_{i-1}}=\NP^{A_{i-1}}$. This equality holds since for all $i\geq 1$, $\NP^{B_i}=\NP^{A_i}$, as one can run the oracle for $A_i$ instead of $B_i$ and negate its answer. Since $A_{i-1}$ is an oracle for $\exists\boldsymbol{\cdot}\forall\boldsymbol{\cdot}\cdots\boldsymbol{\cdot} Q_{i-1}\boldsymbol{\cdot} \class{P}^{\PP}$, the induction hypothesis now implies that
	\begin{equation*}
	\exists\boldsymbol{\cdot}\forall\boldsymbol{\cdot}\cdots\boldsymbol{\cdot} Q_i\boldsymbol{\cdot} \class{P}^{\PP}
	\subseteq \NP^{{\NP_{i-1}}^{{\PP}}} = \NP_i^{{\PP}} = \sigmai^{{\PP}}.
	\end{equation*}
\end{proof}

We can now show the main theorem of this section.
\begin{reptheorem}{thm:QCTODA}
	$\QCPH\subseteq \Po^{\PP^{\PP}}$.
\end{reptheorem}
\begin{proof}
    The claim follows by combining the Cleaning Lemma (Lemma~\ref{l:clean}), Lemma~\ref{l:contain}, and Toda's theorem ($\class{PH}\subseteq \Po^{\PP}$), whose proof relativizes (see, e.g., page 4 of~\cite{F94})).
\end{proof}

\subsection{Detecting non-empty promise gaps is \PP-complete}\label{sscn:ppcomp}

The technique behind the Cleaning Lemma (Lemma~\ref{l:clean}) can essentially be viewed as using a PP oracle to determine whether a given quantum circuit accepts some input with probability within the promise gap $(s,c)$, where $c-s$ is an inverse polynomial. One can ask whether this rather powerful PP oracle can be replaced with a weaker oracle (Remark~\ref{rem:stronger})? We show that unless one deviates from our specific proof approach, the answer is negative. Specifically, we show that the problem of detecting non-empty promise gaps is \PP-complete, {even if} the gap is \emph{constant} in size. Let us begin by formalizing this problem.

\begin{definition}[$\DDET(c,s)$]\label{def:DEVIANT-DETECTION}
    Let $V_n$ be an input for $\BQPC(c,s)$. Then, output \textup{\textsf{YES}} if
    $
        \operatorname{Prob}[V_n\text{ accepts }]\in(s,c),
    $
    and \textup{\textsf{NO}} otherwise.
\end{definition}

We now show that \DDET\ is PP-complete.

\begin{lemma}\label{l:DDETinPP}
    For all $c,s$ with the $c-s$ gap at least an inverse exponential in input size,
    \[
    	\DDET(c,s) \in \PP .
    \]
\end{lemma}

\begin{proof}
	Our approach to show containment in $\PP$ is to give a poly-time many-one reduction from $\DDET(c,s)$ with $c-s$ at least an inverse exponential to $\BQPC(P,Q)$ with $P-Q$ an inverse exponential. (Note that even if $c-s\in \Omega(1)$, we will still have $P-Q$ an inverse exponential.) Let $V_n$ be an input to $\DDET(c,s)$. We construct an instance $V''_n$ of $\BQPC(P,Q)$ as follows.
	
	The first step is to adjust the completeness and soundness parameters for \DDET\ so that they ``straddle'' the midpoint $1/2$. Formally, map $c>s$ to $c'>s'$, respectively, so that $c'(n)-1/2=1/2-s'(n)$. For this, construct the following circuit $V'_n$, whose completeness and soundness parameters we denote by $c'$ and $s'$, respectively.
	
	If $c(n)+s(n)>1$, then with probability
	\begin{equation}
		\alpha:=(c(n)+s(n)-1)/(c(n)+s(n)),
	\end{equation}
	reject, and with probability $1-\alpha$, run $V_n$ and output its answer.
	
	The case of $c(n)+s(n)<1$ is analogous, except with
	\begin{equation}
	\alpha :=(1-c(n)-s(n))/(2-c(n)-s(n)).
	\end{equation}
	Finally, if $c(n)+s(n)=1$, set $c'=c$ and $s'=s$. (Here, we use Observation~\ref{obs:gate}, which allows us to assume $c,s\in\mathbb{Q}$ with $\poly(n)$ bits of precision.)
	
	Next, map $V'_n$ to $V''_n$ as follows: Given a proof $y\in\set{0,1}^m$, (1) run two copies of $V'_n$ in parallel on $y$, (2) negate the output of the second copy of $V'_n$ via a Pauli $X$ gate, (3) apply an AND gate to both output qubits, and (4) measure in the standard basis. Let $p_y$ denote the probability that $V'_n$ accepts $y$. Then, $V_n''$ accepts $y$ with probability $p_y(1-p_y)$.
	
	\emph{Correctness.} Intuitively, since the function $f(x)=x(1-x)$ is maximized over $x\in[0,1]$ when $x=1/2$, the acceptance probability of $V''_n$ is maximized when $y$ falls into the promise gap of $V'_n$, i.e. $p_y\approx 1/2$. Formally, let $c'(n)=1/2+\gamma$ and $s'(n)=1/2-\gamma$ for $\gamma\in(0,1/2]$, and express $p_y=1/2+\delta$ for bias $\delta\in[-1/2,1/2]$. Then, $V''_n$ accepts $y$ with probability
	\begin{equation}
	p_y(1-p_y)=1/4-\delta^2.
	\end{equation}
	It follows that if $p_y\geq c'(n)$ or $p_y\leq s'(n)$, then $V_n''$ accepts $y$ with probability $Q\leq 1/4-\gamma^2$, and if $s'(n)<p_y<c'(n)$, then $V_n''$ accepts $y$ with probability $P > 1/4-\gamma^2$. By Observation~\ref{obs:gate}, we may assume
	\begin{equation}
	P - Q \in \Omega(1/\exp(n)),
	\end{equation}
	thus yielding that a \textsf{YES} instance of $\DDET(c,s)$ with at least inverse exponential $c-s$ is mapped to a \textsf{NO} instance of $\BQPC(P,Q)$ and vice versa with inverse exponential $P-Q$. The claim now follows by Lemma~\ref{l:containedPP}, which says $\BQPC(P,Q)\in\PP$.
\end{proof}

\begin{lemma}\label{l:DDETPPhard}
    There exist $c,s\in \Theta(1)$ such that $\DDET(c,s)$ is $\PP$-hard.
\end{lemma}

\begin{proof}
	Let $\phi:\set{0,1}^n\mapsto\set{0,1}$ be an instance of the PP-complete problem MAJSAT (see proof of Lemma~\ref{l:contain}). We construct an instance $V_n$ of $\DDET(c,s)$ with $c-s\in\Theta(1)$ as follows. Let $V'_n$ be a poly-size quantum circuit which prepares the state
\begin{equation}
	2^{-n/2}\sum_{x\in\set{0,1}^n}\ket{x}_A\ket{\phi(x)}_B\in(\complex^2)^{\otimes n+1},
	\end{equation}
	then measures register $B$ in the standard basis, and accepts if and only if it obtains result $1$. If $\phi$ is a \textsf{YES} instance, then $V'_n$ accepts with probability in range $[1/2+1/2^n,1]$, and if $\phi$ is a \textsf{NO} instance, $V'_n$ accepts with probability in range $[0,1/2]$. Thus, setting (for example) $c=3/4$, $s=1/4$, and constructing circuit $V_n$ which with probability $1/2$ rejects, and with probability $1/2$ runs $V'_n$ and outputs its answer, yields the claim.
\end{proof}

\noindent Lemmas~\ref{l:DDETinPP} and~\ref{l:DDETPPhard} immediately yield the following.
\begin{theorem}\label{thm:ppcomplete}
    There exist $c,s\in \Theta(1)$ such that $\DDET(c,s)$ is $\PP$-complete.
\end{theorem}

\section{Bounding the {complexity} of \texorpdfstring{$\QSi_2$ and $\QSi_3$}{levels in the quantum hierarchy}}
\label{scn:fullyquantum}

In this section, we upper bound the complexity of the second and third levels of our fully quantum hierarchy. For brevity, we sometimes use shorthand $\QSi_2$ and $\QPi_2$ to refer to $\QSi_2(c,s)$ and $\QPi_2(c,s)$, respectively, for completeness and soundness parameters $c$ and $s$, respectively. We begin by restating Theorem~\ref{thm:infEXP} as follows.

\begin{theorem}
\label{thm:inEXP}
For any polynomial $r$ {and input size $n$}, if $c - s \geq 1/{2^{2^{r(n)}}}$, then
\begin{equation}
	\QSi_2(c,s) \subseteq \EXP
	\quad \text{ and } \quad
	\QPi_2(c,s) \subseteq \EXP
\end{equation}
{when $c$ and $s$ are computable in exponential time in $n$.
(Note for classes with small completeness-soundness gaps such as these, a gate set must be fixed\footnote{{The Solovay-Kitaev algorithm (see e.g.~\cite{DN06}) allows one to convert between gate sets in time scaling polylogarithmically in $1/\epsilon$ per gate, where $\epsilon$ is the desired approximation precision per gate. Thus, the setting of doubly exponentially small precision takes superpolynomial overhead to convert between gate sets, which is problematic for promise classes involving \emph{polynomial}-time uniform circuit families (such as $\QSi_2$ and $\QPi_2$). For this reason, in the small gap regime, one can ``circumvent'' the problem by fixing a gate set when defining the class.}}. However, this result is independent of a fixed gate set.)}
\end{theorem}

\begin{proof}
{It suffices to show the first containment for $\QSi_2$, the second containment holds by taking complements and noting that $\coEXP = \EXP$}.

Given a $\QSi_2$ instance, let its two proofs be denoted $\rho_1$ and $\rho_2$, with the former existentially quantified and the latter universally quantified. Let $\alpha$ be the maximum acceptance probability of a $\QSi_2(c, s)$ protocol, i.e. the special case of~(\ref{eq:alpha1}) such that
 \begin{equation}\label{eqn:simpler}
    \alpha := \max_{\rho_1} \min_{\rho_2} \;
    \inner{C}{\rho_1 \otimes \rho_2}
\end{equation}
for accepting POVM operator $C$.
We wish to decide in exponential time whether $\alpha \geq c$ or $\alpha \leq s$. Since the promise gap satisfies $c - s \geq 1/{2^{2^{r(n)}}}$, it suffices to approximate $\alpha$ within additive error (say) $\frac{1}{4}(c-s)$.
Hence, we show how to compute $\gamma\in\real$ such that $|\gamma - \alpha| \leq 1/(4 \cdot {2^{2^{r(n)}}})$ in exponential time.

Beginning with~(\ref{eqn:simpler}), note that we can write $C$ as
\begin{equation}
C = \tr_{\rm anc} \left[ (I\otimes \ketbra{0\cdots 0}{0 \cdots 0}_{\rm anc})V^\dagger_n(\ketbra{1}{1}_{\rm out}\otimes I)V_n (I\otimes \ketbra{0\cdots 0}{0 \cdots 0}_{\rm anc}) \right]
\end{equation}
for verification circuit $V_n$. By definition, $V_n$ is generated by a polynomial-time Turing machine, which we assume specifies $V_n$ via a sequence of gates from a  universal gate set $G$ (e.g., $\{$CNOT, H, T$\}$). Since we wish to proceed via numerical optimization techniques, we begin by computing a numerical approximation $C'$ to $C$. Specifically, in exponential time,
we can approximate each entry\footnote{This can be accomplished in exponential time as follows. Replace gate set $G$ with $G'$ by approximating each entry of each gate in $G$ using $2^{s(n)}$ bits of precision, for some sufficiently large, fixed polynomial $s$. Define $C'$ as $C$, except each use of a gate $U\in G$ is replaced with its approximation $U'\in G'$. Then, via the well-known bound $\snorm{U_m\cdots U_1-V_m\cdots V_1}\leq\sum_{i=1}^m\snorm{U_i-V_i}$ (for unitary $U_i,V_i$), it follows that $\snorm{C'-C}\in O(\poly(n)/(2^{2^{s(n)}})$, since $V_n$ contains $\poly(n)$ gates. Here, $\snorm{A}=\max_{\ket{\psi}} \norm{A\ket{\psi}}_2$ for unit vectors $\ket{\psi}$ denotes the spectral or operator norm. Finally, apply the fact that $\max_{i,j}\abs{A(i,j)}\leq \snorm{A}$ (p. 314 of~\cite{HJ90}).}
of $C$ using $2^{q(n)}$ bits of precision, for some polynomial $q$. Therefore, we have
\begin{equation} \label{bound1}
|\inner{C - C'}{\sigma_1 \otimes \sigma_2}|
\leq
\| C - C' \|_2 \| \sigma_1 \otimes \sigma_2 \|_2
\leq
\| C - C' \|_2
=
O\left(2^{2p(n)-2^{q(n)}}\right)
\end{equation}
for any density matrices $\sigma_1$, $\sigma_2$. (Recall $p(n)$ is the size of each proof, for some polynomial $p$.)
 Therefore, for sufficiently large polynomial $q$, we have that

\begin{equation}
	\alpha' := \max_{\rho_1} \ \min_{\rho_2} \ \{ \inner{C'}{\rho_1 \otimes \rho_2} : \tr(\rho_1) = \tr(\rho_2) = 1, \; \rho_1, \rho_2 \succeq 0 \}, \label{eq:approxalpha1}
	\end{equation}
satisfies $\abs{\alpha - \alpha'} \leq \frac{1}{8} \cdot 2^{-2^{r(n)}} \leq \frac{1}{8}(c-s)$.

	We now use SDP duality (in a manner reminiscent of LP solutions for the Chebyshev approximation problem, p.~293 of~\cite{BV04}) to rephrase \eqref{eq:approxalpha1} as an SDP. Suppose we fix a feasible $\rho_1$ and solve the inner optimization problem in~\eqref{eq:approxalpha1}. Then:
	\begin{equation}
	\alpha'(\rho_1) := \min_{\rho_2} \ \{ \inner{C'}{\rho_1 \otimes \rho_2} : \tr(\rho_2) = 1, \; \rho_2 \succeq 0 \}.
	\end{equation}
	We can rewrite
	$
	\inner{C'}{\rho_1 \otimes \rho_2}$
	as
	$\inner{\tr_{1}[(\rho_1 \otimes I)C']}{\rho_2}
	$
	where $\tr_1$ is the partial trace over the register that $\rho_1$ acts on. Additionally, as
	$ \tr_{1}[(\rho_1 \otimes I)C'] = \tr_{1}[(\rho_1^{1/2} \otimes I)C(\rho_1^{1/2} \otimes I)] $, this term
	is Hermitian and positive semidefinite.
	This implies that the best choice for $\rho_2$ is a rank-$1$ projector onto the eigenspace corresponding to the least eigenvalue.
	In other words,
	$\alpha'(\rho_1) = \lambda_{\min}(\tr_{1}[(\rho_1 \otimes I)C'])$
	where $\lambda_{\min}(X)$ denotes the least eigenvalue of the operator $X$.
	For fixed $\rho_1$, this minimum eigenvalue calculation can be rephrased via the dual optimization program for $\alpha'(\rho_1)$,
	\begin{equation}
	\alpha'(\rho_1) = \max_t \ \{ t : t I \preceq \tr_{1}[(\rho_1 \otimes I)C'] \}.
	\end{equation}
	Re-introducing the maximization over $\rho_1$, we hence obtain
	\begin{equation}
	\alpha' = \max_{\rho_1, t} \ \{ t : t I \preceq \tr_{1}[(\rho_1 \otimes I) C'], \; \tr(\rho_1) = 1, \; \rho_1 \succeq 0 \},
	\end{equation}
	which is a semidefinite program.
	
	With an SDP in hand, we now apply the Ellipsoid Method to obtain an estimate, $\gamma$, for $\alpha'$. Note that not all SDPs can be solved in polynomial time, as the runtime of the Ellipsoid Method depends in part on two parameters, $R$ and $\epsilon$, where $R$ is the radius of a ball (with respect to the Euclidean norm) containing the feasible region, and $\epsilon$ is the radius of a ball contained in the feasible region (see \cite{GLS93} for details). For this reason, we give an equivalent SDP which allows us to bound $R$ and $\epsilon$ as follows. First, relax the constraint ${\tr(\rho) = 1}$ to $\tr(\rho) \leq 1$. Second, replace $t$ with $t_1 - t_2$ where $t_1, t_2 \geq 0$. From context, we know $t$ is a probability, and so we have the implicit constraint $t \in [0,1]$. {Therefore}, we add redundant constraints $t_1, t_2 \leq 100$ without changing $\alpha'$. Thus, we have the following reformulation of $\alpha'$.
	\begin{equation}
	\alpha' = \max_{\rho_1, t_1, t_2} \ \{ t_1 - t_2 : (t_1 - t_2) I \preceq \tr_{1}[(\rho_1 \otimes I) C'], \; \tr(\rho_1) \leq 1, \; \rho_1 \succeq 0, \ t_1, t_2 \in [0, 100] \}. \label{eq:alpha2}
	\end{equation}
	
	We can now use the Ellipsoid Method to approximately solve this SDP in time that is exponential in $n$. We follow a similar analysis to~\cite{W09} and find a $\gamma$ such that $\abs{\gamma - \alpha'} \leq \epsilon$ in time
	\begin{equation} \label{runtime}
	\poly(\log(R),\log(1/\epsilon), n', m, J),
	\end{equation}
	for parameters $R$, $\epsilon$, $n'$, $m$, and $J$ defined as:
	\begin{itemize}
		\item $R$: This is equal to the maximum of $\| \rho_1 \oplus t_1 \oplus t_2 \|_2$ over all feasible $(\rho_1, t_1, t_2)$. Since we have
\begin{equation}
\| \rho_1 \oplus t_1 \oplus t_2 \|_2
\leq
\| \rho_1 \oplus t_1 \oplus t_2 \|_1
=
\tr(\rho_1 \oplus t_1 \oplus t_2)
=
\tr(\rho_1) + t_1 + t_2
\leq 201,
\end{equation}
for feasible $(\rho_1, t_1, t_2)$,
we can set $R = 201$.		

		\item $\epsilon$: This is the radius of a small ball contained in the feasible region. Specifically, $\epsilon$ is defined so that there exists feasible $(\rho_1, t_1, t_2)$ such that $(\rho_1, t_1, t_2) + (\sigma, \tau_1, \tau_2)$ is feasible for all $(\sigma, \tau_1, \tau_2)$ with $\| \sigma \oplus  \tau_1 \oplus \tau_2 \|_2 \leq \epsilon$.
Since we have 		
\begin{equation}
\epsilon \geq
\| \sigma \oplus \tau_1 \oplus \tau_2 \|_2
=
\|\sigma \|_2 + |\tau_1| + |\tau_2|
\geq
\max \{ \|\sigma \|_2, |\tau_1|, |\tau_2| \},
\end{equation}
we will use the more convenient bound
$\max \{ \|\sigma \|_2, \tau_1, \tau_2 \} \leq \epsilon$ for the analysis.
We choose the interior point
\begin{equation}
\rho_1 = \frac{1}{\dim(\rho_1)^2} I,
\quad
t_1 = 10,
\quad
t_2 = 20,
\end{equation}
and
\begin{equation}
\epsilon = \frac{1}{8} \cdot 2^{(-2^{r(n)})} \leq \frac{1}{8}(c-s).
\end{equation}
Note that this has the sufficiently small accuracy we require.

We now prove that $(\rho_1, t_1, t_2) + (\sigma, \tau_1, \tau_2)$ is feasible so long as
$\max \{ \| \sigma \|_2, \tau_1, \tau_2 \} \leq \epsilon$. One can check that for these values, we have
\begin{enumerate}
\item $\rho_1 + \sigma \succeq \rho_1 - \| \sigma \|_{\infty} I \succeq \rho_1 - \epsilon I \succeq 0$, \; (where we used $\| \sigma \|_2 \leq \epsilon$ implies $\| \sigma \|_{\infty} \leq \epsilon$),
\item $\tr(\rho_1 + \sigma)
	= \tr(\rho_1) + \tr(\sigma)
	\leq \tr(\rho_1) + \| \sigma \|_2 \| I \|_2
	\leq \frac{1}{\dim(\rho_1)} + \epsilon \sqrt{\dim(\rho_1)}
	\leq 1$
, \; (using the Cauchy-Schwarz inequality),
\item $t_1 + \tau_1 \in [0, 100]$ and $t_2 + \tau_2 \in [0, 100]$, \;
\item $( (t_1 + \tau_1) - (t_2 + \tau_2)) I \prec 0 \preceq \tr_{1} [ ( (\rho_1 + \sigma) \otimes I) C']$, \; (since $(t_1 + \tau_1) - (t_2 + \tau_2) < 0$ and $\rho_1 + \sigma \succeq 0$ as shown above).
\end{enumerate}
		\item $n'$: The dimension of $\rho_1 \oplus t_1 \oplus t_2$, which is equal to the sum of the dimensions, i.e.,  		
			$O(2^{p(n)})$.
		\item $m$: The dimension of the operators appearing in the constraints. Note, from~\eqref{eq:alpha2}, that constraint ${(t_1 - t_2) I \preceq \tr_{1}[(\rho_1 \otimes I)C]}$
		involves operators acting on a space of dimension $O(2^{p(n)})$. Moreover, there are only $3$ other inequality constraints: ${\tr(\rho) \leq 1}$ and $t_1, t_2 \leq 100$. Thus, $m = O(2^{p(n)})$.
		\item $J$: The maximum bit-length of the entries in $C'$,
		{which is $2^{q(n)}$, by definition}.
	\end{itemize}

	We conclude that the $\EXP$ protocol approximates $\alpha'$ via $\gamma$, which is correct up to an additive error of $\frac{1}{8}(c-s)$. Finally, if $\gamma \geq (c+s)/2$, we output YES (i.e. $x \in A_{\yes}$). Otherwise, we output NO (i.e. $x \in A_{\no}$).
\end{proof}

Using the power of non-determinism, we can also bound the {complexity} of $\QSi_3$ and $\QPi_3$.

\begin{theorem}\label{thm:inNEXP}
For any polynomial $r$ and input size $n$, if $c-s\geq 1/r(n)$, then
\begin{equation}
	\QMA(2) \subseteq \QSi_3 \subseteq \NEXP
	\quad \text{ and } \quad
	\mathrm{co} \textrm{-} \QMA(2) \subseteq \QPi_3 \subseteq \mathrm{co} \textrm{-} \NEXP,
\end{equation}
where all classes have completeness and soundness $c$ and $s$, respectively. Moreover, if we allow smaller gaps ({in principle, gaps which are at most inverse singly exponential in $n$ suffice}), such as $c - s \geq 1/{2^{2^{r(n)}}}$, then
\begin{equation}
	\QMA(2)(c,s) = \QSi_3(c,s) = \NEXP
	\quad \text{ and } \quad
	\mathrm{co} \textrm{-} \QMA(2) = \QPi_3(c,s) = \mathrm{co} \textrm{-} \NEXP.
\end{equation}
Here, we assume $c$ and $s$ are computable in {exponential time in $n$.
(Note for classes with doubly exponentially small completeness-soundness gaps, a gate set must be fixed. However, this result is independent of the choice of a fixed gate set.)}
\end{theorem}

\begin{proof}
{Again, we prove the statements involving $\QSi_3$ and the analogous statements involving $\QPi_3$ follow by taking complements.}

Consider the maximum acceptance probability of a $\QSi_3$ protocol,
\begin{equation}
\beta := \max_{\rho_1} \ \min_{\rho_2} \ \max_{\rho_3} \ \{ \inner{C}{\rho_1 \otimes \rho_2 \otimes \rho_3} : \tr(\rho_1) = \tr(\rho_2) = \tr(\rho_3) = 1, \; \rho_1, \rho_2, \rho_3 \succeq 0 \}
\end{equation}
where $C$ is the POVM element corresponding to the verifier accepting. 	
As in the proof of Theorem~\ref{thm:inEXP},
define $C'$ to be equal to $C$ where each entry is correct up to $2^{q(n)}$ bits of precision. Consider now the optimization problem
\begin{equation}
\beta' := \max_{\rho_1} \ \min_{\rho_2} \ \max_{\rho_3} \ \{ \inner{C'}{\rho_1 \otimes \rho_2 \otimes \rho_3} : \tr(\rho_1) = \tr(\rho_2) = \tr(\rho_3) = 1, \; \rho_1, \rho_2, \rho_3 \succeq 0 \}
\end{equation}
and note that $| \beta' - \beta | = O(2^{3p(n)-2^{q(n)}})$ by an argument similar to \eqref{bound1}.
	
Now, suppose we non-deterministically guess a value for the optimal $\rho_1^{\star}$ (which, recall, exists by compactness) by a matrix, each of whose entries is specified up to $2^{t(n)}$ bits of precision for some polynomial $t$. Call this approximation $\rho_1'$. We now consider the optimization problem
\begin{align}
\beta'' & := \min_{\rho_2} \ \max_{\rho_3} \ \{ \inner{C'}{\rho_1' \otimes \rho_2 \otimes \rho_3} : \tr(\rho_2) = \tr(\rho_3) = 1, \; \rho_2, \rho_3 \succeq 0 \} \\
& = \min_{\rho_2} \ \max_{\rho_3} \ \{ \inner{C''}{\rho_2 \otimes \rho_3} : \tr(\rho_2) = \tr(\rho_3) = 1, \; \rho_2, \rho_3 \succeq 0 \}
\end{align}
where $C'' := \tr_1 [(\rho_1' \otimes I \otimes I)C']$
is the matrix which hardcodes $\rho'_1$ into $C'$.
We now bound $| \beta' - \beta'' |$.
For any density operators $\sigma_2$ and $\sigma_3$, we have
\begin{align}
|\inner{C'}{(\rho_1^{\star} - \rho_1') \otimes \sigma_2 \otimes \sigma_3}|
& \leq
\| C' \|_2 \| \rho_1^{\star} - \rho_1' \|_2 \\
\| C' \|_2
& \leq
\| C' - C \|_2 + \| C \|_2
\leq
2 \| I \|_2
=
O(2^{\frac{3 p(n)}{2}}) \\
\| \rho_1^{\star} - \rho_1' \|_2
& = O(2^{p(n)-2^{t(n)}}).
\end{align}
Thus we have
$| \beta' - \beta'' | = O(2^{{\frac{5 p(n)}{2}}-2^{t(n)}})$ and therefore we can choose $q$ and $t$ such that
\begin{equation}
| \beta - \beta'' | \leq | \beta - \beta' | + | \beta' - \beta''| \leq \frac{1}{8} \cdot 2^{-2^{r(n)}} \leq \frac{1}{8} (c-s)
\end{equation}
as before.

Since $C''$ still has exponential bit-length, is positive semidefinite, and can be computed in non-deterministic exponential time, we can repeat the arguments from Theorem~\ref{thm:inEXP} to find $\gamma$ in exponential time such that $|\beta'' - \gamma| \leq \frac{1}{8} (c-s)$.
Thus, if $\gamma \geq \frac{1}{2}(c+s)$ we output YES (i.e., $x \in \ayes$). Otherwise, we output NO (i.e., $x \in \ano$).
This yields
	\begin{equation} \label{NEXPlemma}
	\QSi_3 \subseteq \NEXP,
	\end{equation}
	even when $\QSi_3$ has small gap.
	The rest of the theorem holds by a result of
	Pereszl\'{e}nyi~\cite{attila}, who proved that the equality
	$
		\QMA(2)(c,s) = \NEXP
	$
	holds {when $c - s \geq 1/2^{2^r}$ for polynomial $r$}.
	Combining this with the fact that $\QMA(2) \subseteq \QSi_3$ and \eqref{NEXPlemma} finishes the proof.
\end{proof}

\section{A Karp-Lipton type theorem}
\label{sec:karp-lipton}

The Karp-Lipton~\cite{KL80} theorem showed that if $\NP\subseteq \Ppoly$ (i.e. if \NP\ can be solved by polynomial-size non-uniform circuits), then $\sigmatwo=\pit$ (which in turn collapses PH collapses to its second level). Then, building on the conjecture that the polynomial hierarchy is infinite, this result implies that $\NP \not\subset \Ppoly$ (a stronger claim than $\Po \neq \NP$ as $\Po \subseteq \Ppoly$). Some attempts to separate $\NP$ from $\Po$ use this as a basis to try and prove the stronger claim instead. For instance, this has lead to the approach of proving super-polynomial circuit lower bounds for circuits of \NP-complete problems. Here, we show that the proof technique used by Karp and Lipton carries over directly to the quantum setting, \emph{provided} one uses the stronger hypothesis
\begin{equation}
	\PrQCMA\subseteq \Bpoly
\end{equation}
(as opposed to $\QCMA\subseteq \Bpoly$). Whether this causes \QCPH\ to collapse to its second level, however, remains open (see Remark~\ref{rem:collapse} below). We begin by formally defining \Bpoly\ and \PrQCMA.

\begin{definition}[$\Bpoly$] \label{def:Bpoly}
	A promise problem $\Pi=(\ayes,\ano)$ is in $\Bpoly$ if there exists a polynomial-sized family of quantum circuits
	$\{C_n\}_{n \in \mathbb{N}}$ and a collection of binary advice strings $\{a_n\}_{n \in \mathbb{N}}$ with $|a_n| = \poly(n)$, such that for all $n$ and all strings $x$ where $|x| = n$, $\Pr[C_n(\ket{x}, \ket{a_n}) = 1] \geq 2/3$ if $x \in \ayes$ and $\Pr[C_n(\ket{x}, \ket{a_n}) = 1] \leq 1/3$ if $x \in \ano$.
\end{definition}

\noindent Equivalently, $\Bpoly$ is the set of promise problems solvable by a \emph{non-uniform} family of poly-sized bounded error quantum circuits. It is used as a quantum analogue for $\Ppoly$ in this scenario. Here, we remark on the use of $\mpoly$ instead of $\poly$ in Definition~\ref{def:Bpoly}. Note that $\BQP_{/\poly}$ accepts Karp-Lipton style advice i.e. it is a $\BQP$ circuit that accepts a poly-sized advice string to provide \emph{some answer} with probability at least $2/3$ even if the ``advice is bad''. On the other hand, $\Bpoly$ accepts Merlin style advice i.e. it is a $\BQP$ circuit accepting poly-sized classical advice such that the output is correct with probability at least $2/3$ if the ``advice is good''. Note $\BQP_{/\poly}$ versus $\Bpoly$ is analogous to the ``$\exists\boldsymbol{\cdot}\BPP$ versus \MA'' phenomenon. Moreover, as we are concerned with variations of $\QCMA$, and not $\exists\boldsymbol{\cdot}\BQP$, $\Bpoly$ is the right candidate for us.

\begin{definition}[$\PrQCMA$]\label{def:PrQCMA}
    A promise problem $A=(\ayes,\ano)$ is said to be in $\PrQCMA(c,s)$ for polynomial-time computable functions $c,s:\natural\mapsto[0,1]$ if there exists polynomially bounded functions $p,q:\natural\mapsto\natural$ such that $\forall \ell \in \natural, \ c(\ell) - s(\ell) \geq 2^{-q(\ell)}$, and there exists a polynomial-time uniform family of quantum circuits $\{V_n\}_{n \in \natural}$ that takes a classical proof ${y}\in \set{0,1}^{p(n)}$ and outputs a single qubit. Moreover, for an $n$-bit input $x$:
    \begin{itemize}
     \item Completeness: If $x\in \ayes$, then $\exists\ y$ such that $V_n$ accepts $y$ with probability at least $c(n)$.
     \item Soundness: If $x\in \ano$, then $\forall\ y$, ${V_x}$ accepts $y$ with probability at most $s(n)$.
    \end{itemize}
    Define $\PrQCMA=\bigcup_{c,s} \PrQCMA(c,s)$.
\end{definition}

\begin{obs}
	\label{obs:prqcma1}
	The proof of Theorem~\ref{QCMAPC} and footnote 2 in~\cite{JKNN12} show that by choosing an appropriate universal gate set (e.g. Hadamard, Toffoli, NOT), one has that
	\begin{equation}
		\PrQCMA=\PrQCMA(1,1-1/\exp(n)).
	\end{equation}
\end{obs}

As an aside, note that $\QCMA$ is defined with $c - s \in \Omega(1/\poly(n))$. Recall from the discussion in Section~\ref{sec:results} that the main obstacle to the recursive arguments that work well for \NP\ in~\cite{KL80} is the ``promise problem'' nature of $\QCPi_2$ and \QCMA. However, exploiting the perfect completeness of \PrQCMA\
and the fact that $\forall s\leq s' <c, \; \PrQCMA(c, s) \subseteq \PrQCMA(c, s')$, we "recover'' the notion of a decision problem in a rigorous sense by working with \PrQCMA\ as demonstrated below.

\begin{claim}\label{clm:decision}
For every promise problem $\Pi' = (\ayes, \ano) \in \PrQCMA(c, s)$ with verifier $V'$, there exists a verifier $V$ (a poly-time uniform quantum circuit family), a polynomial $q$ and a decision problem $\Pi=(\ayes,\set{0,1}^*\setminus\ayes)$ such that $\Pi \in \PrQCMA(1,1-2^{-q(n)})$ with verifier $V$. Moreover, for all proofs $y$, $V$ accepts $y$ with probability either $1$ or at most $1-2^{-q(n)}$.
\end{claim}

\begin{proof}By Observation~\ref{obs:prqcma1}, we may assume
	\begin{equation}
	\PrQCMA(c,s)=\PrQCMA(1,1-2^{-p(n)})
	\end{equation}
	for some polynomial $p$. Let $\Pi = (\ayes, \ano) \in \PrQCMA(1, 1-2^{-p(n)})$ be a promise problem with verifier $V$. The concern is that for $x\in \ayes$, there may exist a proof $y$ accepted by $V$ with probability in $(1,1-2^{-p(n)})$. By Observation~\ref{obs:gate}, we may assume the acceptance probabilities of $V$ are integer multiples of $2^{-q(n)}$ for some polynomial $q$. Since $p$ and $q$ are polynomials, there exists $n_0\geq 0$ such that $\forall n\geq n_0$, either $p(n)\geq q(n)$ or vice versa. Thus, updating the soundness parameter to $1-2^{-p(n)}$ in the former case and to $1-2^{-q(n)}$ in the latter case ensures that no proofs are accepted by $V$ in the promise gap for sufficiently large $n$. This yields the second claim of the observation. The first claim now also follows, since if no proofs are accepted in the promise gap, then certainly the optimal proof is also not accepted in the gap.
\end{proof}

Note that the same process fails to map a promise problem $\Pi' \in \QCMA(1, s)$ to a corresponding decision problem $\Pi \in \QCMA(1, s')$ where $ 1 < s' \leq s$. As shown above, $s'$ could very well be exponentially close to $1$, which would violate the requirement, by definition, for $\QCMA$ that the promise gap should be an inverse polynomial function in the input size.

Building on this "decision problem'' flavour of $\PrQCMA$, we first show:

\begin{lemma}
\label{lem:ckt-eq}
Suppose $\PrQCMA \subseteq \Bpoly$. Then, for every promise problem $\Pi=(\ayes,\ano)$ in $\PrQCMA$ and every $n$-bit input $x$, there exists a polynomially bounded function $p:\natural \mapsto \natural$ and a bounded error polynomial time non-uniform quantum circuit family $\set{C_n}_{n \in \natural}$ such that:
\begin{itemize}
\item {if $x \in \ayes$, then $C_n$ outputs valid proof $y \in \set{0, 1}^{p(n)}$ such that $(x, y)$ is accepted by the corresponding $\PrQCMA$ verifier with probability $1$;}
\item {if $x \in \ano$, then $C_n$ outputs a symbol $\bot$ with probability exponentially close to $1$ signifying that there is no $y \in \set{0, 1}^{p(n)}$, such that $(x, y)$ is accepted by the corresponding $\PrQCMA$ verifier with probability $1$.}
\end{itemize}
\end{lemma}

\begin{proof}
	To begin, recall from Claim~\ref{clm:decision} that we may assume that a given promise problem $\Pi$ in \PrQCMA\ has (a) completeness/soundness parameters $(1,1-2^{-q(n)})$ for a polynomial $q$ and (b) a verifier $V_n$ for an $n$-bit input $x$ which accepts no proofs with probability in the promise gap. Since
	$\PrQCMA\subseteq \Bpoly$,
	by assumption, there exists a non-uniform polynomial-size quantum circuit family $\set{C'_n}_{n \in \natural}$ that accepts $x$ as advice such that for any $x\in\ayes$, $C'_n$ accepts with probability at least $2/3$ and rejects with probability $2/3$ otherwise. By using standard parallel repetition, we may assume without loss of generality that $C'_n$ accepts or rejects the corresponding cases with probability at least $1-2^{-p(n)}$ for some polynomial $p$. Now, we would like to construct a \Bpoly\ circuit $C_n$ that uses the input $x$ and a description of $C'_n$ as polynomial-sized advice and outputs a valid proof $y$ such that $V_n$ accepts $(x, y)$ with probability $1$.
	
	The construction of $\set{C_n}$ is now as follows. $C_n$ first runs $C'_n$ using $x$ to check if $x\in \ayes$; if not, it rejects and outputs $\bot$. To find a proof $y$, we now use standard self-reducibility ideas from SAT, coupled with the crucial Claim~\ref{clm:decision}. Specifically, fix $y_{1} = 0$ (i.e. the first bit of $y$) to obtain a new circuit $C'_{n,1}$, run $C'_{n,1}$ on $(x,y_1)$ and record its answer $z_1\in\set{0,1}$. Since no proofs are accepted in the gap as per Claim~\ref{clm:decision}, $C'_{n,1}$ is a valid \PrQCMA\ machine (i.e. satisfying the promise of the completeness/soundness parameters). Thus, with high probability, if $z_1=1$ there is an accepting proof for $x$ whose first bit is $0$ and if $z_1=0$, there is a proof with the first bit set to $1$. Hence, we can fix $y_1$'s accordingly. Iterating this process successively for all remaining bits of $y$ yields the claim.
\end{proof}

We next give a quantum-classical analogue of the Karp-Lipton theorem.

\begin{reptheorem}{thm:QCKL}
[A  Quantum-Classical Karp-Lipton Theorem]
If $\PrQCMA \subseteq \Bpoly$ then $\QCPi_2 = \QCSi_2$.
\end{reptheorem}

\begin{proof}
	We essentially follow the proof of the original Karp-Lipton theorem, coupled with careful use of Observation~\ref{obs:gate}. To show $\QCPi_2 = \QCSi_2$, it suffices to show that $\QCPi_2 \subseteq \QCSi_2$. To see this, consider promise problem $A \in \QCSi_2$. Now, $\bar{A}$ (the complement of $A$) is in $\QCPi_2$ by definition. However, if $\QCPi_2 \subseteq \QCSi_2$, then $\bar{A} \in \QCSi_2$,  which in turn implies by definition that $A \in \QCPi_2$, as desired.
	
	To show $\QCPi_2 \subseteq \QCSi_2$, let $A = (\ayes, \ano)$ be a $\QCPi_2$ problem.
	As $\QCPi_{2}$ has perfect completeness from Result~\ref{res:err}, there exist polynomials $p, r$ and a polynomial time uniform family of quantum circuits $\{V_i\}_{i \in \mathbb{N}}$ that take as input a string $x \in \{0, 1\}^{n}$ for some $n \in \natural$, two classical proofs $u, v \in \{0, 1\}^{p(n)}$, and outputs a single qubit such that:
	\begin{align}
	x \in \ayes \quad & \mathrel{{\Rightarrow}} \quad \forall u \; \exists v  \; \pr[V_{n}(x, u, v) = 1] = 1, \label{eq:3c} \\
	x \in \ano \quad & \mathrel{{\Rightarrow}} \quad \exists u \; \forall v \;  \pr[V_{n}(x, u, v) = 1] \leq \frac{1}{2^{r(n)}}. \label{eq:4c}
	\end{align}
	
	Let us now highlight the difficulty in proving the claim for \QCMA\ instead of \PrQCMA. Specifically, if we fix the first proof $u$, what we would ideally require is that the resulting existentially quantified computation over $v$, denoted $M_u$\footnote{{Notice that $M_u$ is the remnant of the verifier circuit obtained from $V_n$ when the first proof register is loaded with $u$. In a slight abuse of notation, $M_u$ will be referred to both as a computation and as a circuit.}}, is in \QCMA. Indeed, if $x\in\ayes$, then for any fixed $u$, there exists a $v$ causing $M_u$ to accept with certainty. The problem arises when $x\in \ano$, in which case, we require that for all $v$, $M_u$ accepts with probability at most $s$ for some soundness parameter $s$ inverse polynomially gapped away from $1$. Unfortunately, the definition of $\QCPi_2$ only ensures this holds for \emph{some} $u$, and not necessarily \emph{all} $u$. To circumvent this, we use  Observation~\ref{obs:gate}, which implies we may assume $M_u$'s acceptance probabilities are given by rational numbers with $\poly(n)$ bits of precision (assuming an appropriate universal gate set is used). It follows that if $M_u$ does not accept some $v$ with probability $1$, then it must reject $v$ with probability at least $1-2^{-q(n)}$ for some efficiently computable polynomial $q$. Thus, by definition $M_u$ is a $\PrQCMA(1,1-2^{-q(n)})$ computation, to which we may now apply our hypothesis that $\PrQCMA\subseteq \Bpoly$. (Note: There is a subtle point here --- the precise choice of $q$ depends on the length of circuit $M_u$, which in turn depends on the Hamming weight of $u$, since we can simulate ``fixing'' $u$ by adding appropriate Pauli $X$ gates to our circuit. Nevertheless, it is trivial to choose a polynomial $q$ which provides sufficient precision in our rational approximation in order to accommodate the fixing of proofs $u$ of any Hamming weight.)
	
	Since for any fixed $u$, $M_u$ denotes a $\PrQCMA$ computation, our assumption says that there exists a non-uniform family of polynomial-sized bounded-error quantum circuits $\set{Q'_n}_{n \in \natural}$ that accepts $(x, u)$ and a description of $M_u$ as advice and outputs a bit such that:
	\begin{itemize}
		\item if there exists a proof $v$ such that $M_u$ accepts $(x,v)$ with probability $1$, then $Q'_n$ accepts $(x, u, M_u)$ with probability at least $2/3$, and
		\item if for all proofs $v$, $M_u$ accepts $(x,v)$ with probability at most $1-2^{-q(n)}$, then $Q'_n$ accepts $(x, u, M_u)$ with probability at most $1/3$.
	\end{itemize}
	\noindent Crucially, the set $\set{Q'_n}$ is non-uniform, and thus $Q'_n$ depends only on $n$, not the choice of $x$ or $u$.
	
	Continuing, from Lemma~\ref{lem:ckt-eq}, we now conclude there exists a bounded-error polynomial time non-uniform quantum circuit family $\set{Q_n}_{n \in \natural}$ which, whenever $x\in\ayes$, outputs a proof $v$ which $M_u$ accepts with certainty. For clarity, note that $Q_n$ accepts $(x, u)$ and a description of $M_u$ as advice and outputs a string $v \in \set{0, 1}^{p(n)}$. Suppose $Q_n$ outputs the correct answer with probability at least $1-2^{-s(n)}$ for some polynomial $s$, as per Lemma~\ref{lem:ckt-eq}. Using the existence and non-uniformity of $\set{Q_n}$, as done in the proof of the classical Karp-Lipton theorem, we claim we may now swap the order of the quantifiers and write:
	\begin{itemize}
		\item If $x\in\ayes$, then $\exists Q_n\;\forall u \pr[V_{n}(x, u, Q_n(x,u, M_u)) = 1] \geq 1-2^{-s(n)}$, and
		\item if $x\in\ano$, then $\forall Q_n\;\exists u \pr[V_{n}(x, u, Q_n(x, u, M_u)) = 1] \leq \frac{1}{2^{r(n)}}$.
	\end{itemize}
	This would imply the desired claim that $\QCPi_2 \subseteq \QCSi_2$.
	
	To see that we may indeed swap quantifiers in this fashion, assume first that $x\in\ayes$. Then, choosing the non-uniform circuit family from Lemma~\ref{lem:ckt-eq} yields that for any fixed $x$ and $u$, with probability at least $1-2^{-s(n)}$, $Q_n$ outputs a proof $v$ such that $C_n$ accepts $(x,u,v)$ with probability $1$. Conversely, if $x\in \ano$, since for an appropriate choice of $u$, there are \emph{no} proofs $v$ such that $C_n$ accepts $(x,u,v)$ with probability more than $2^{-r(n)}$. Then, clearly no choice of $Q_n$ is able to generate a proof $Q_n(x,u, M_u)$ such that $C_{n}$ accepts $(x, u, Q_n(x,u, M_u))$ with probability more than $2^{-r(n)}$.
\end{proof}

\begin{remark}[\textbf{Collapse of \QCPH?}]\label{rem:collapse}
An appeal of the classical Karp-Lipton theorem is that it implies that if $\NP\subseteq\Ppoly$, then PH collapses to its second level; this is because if $\pit=\sigmatwo$, then \PH\ collapses to $\sigmatwo$. Does an analogous statement hold for \QCPH\ as a result of Theorem~\ref{thm:QCKL}? Unfortunately, the answer is not clear. The problem is similar to that outlined in Remark~\ref{rem:langvsprom}. Namely, classically $\pit=\sigmatwo$ collapses \PH\ since for any $\Pi_3$ decision problem, fixing the first (universally) quantified proof yields a $\sigmatwo$ computation. But this can be replaced with a $\pit$ computation by assumption, yielding a computation with quantifiers $\forall\forall\exists$, which trivially collapses to $\forall\exists$, i.e. $\Pi_3\subseteq \pit$. In contrast, for (say) $\QCPi_3$, similar to the phenomenon in Remark~\ref{rem:langvsprom}, fixing the first (universally) quantified proof does \emph{not} necessarily yield a $\QCSi_2$ computation. Thus, a recursive application of the assumption $\QCSi_2=\QCPi_2$ cannot straightforwardly be applied.
\end{remark}

Since $\PrQCMA$ plays an important role in Theorem~\ref{thm:QCKL}, we close with an upper bound on $\PrQCMA$.

\begin{lemma}\label{lem:bnd}
$\PrQCMA \subseteq \NP^{\PP}$.
\end{lemma}

\begin{proof}
	Let $V$ be a $\PrQCMA$ verifier. Using Claim~\ref{clm:decision}, we may assume that for any proof $y$, $V$ either accepts $y$ with probability $1$ or rejects with probability at most $1-2^{-q(n)}$. Thus, for any fixed $y$, the resulting computation $V_y$ is a $\PrBQP$ computation. This implies $\PrQCMA\subseteq \exists\cdot\PrBQP$ (see also Remark~\ref{rem:langvsprom}). But by Definition~\ref{def:exists},  $\exists\cdot\PrBQP\subseteq \NP^{\PrBQP}$. Combining this with Lemma~\ref{l:containedPP}, which says $\PrBQP\subseteq\PP$, yields the claim.
\end{proof}

\section*{Acknowledgements}
SG and AS thank the Centre for Quantum Technologies at the National University of Singapore for their support and hospitality, where part of this research was carried out. SG acknowledges support from NSF grants CCF-1526189 and CCF-1617710. AS is supported by the Department of Defense. JY was supported by a Virginia Commonwealth University Presidential Scholarship. Research at the Centre for Quantum Technologies is partially funded by the Singapore Ministry of Education and the National Research Foundation under grant R-710-000-012-135. This research was supported in part by the QuantERA ERA-NET Cofund project QuantAlgo.
This research was supported in part by Perimeter Institute for Theoretical Physics.
Research at Perimeter Institute is supported by the Government of Canada through
the Department of Innovation, Science and Economic Development Canada and by
the Province of Ontario through the Ministry of Research, Innovation and Science.

\bibliographystyle{alpha}
\bibliography{gsssy_bib}

\newcommand{\etalchar}[1]{$^{#1}$}
\begin{thebibliography}{ABOBS08}

\bibitem[ABD{\etalchar{+}}09]{ABDFS09}
S.~Aaronson, S.~Beigi, A.~Drucker, B.~Fefferman, and P.~Shor.
\newblock The power of unentanglement.
\newblock {\em Theory of Computing}, 5:1--42, 2009.

\bibitem[ABOBS08]{ABBS08}
D.~Aharonov, M.~Ben-Or, F.~Brand{\~{a}}o, and O.~Sattath.
\newblock The pursuit for uniqueness: {E}xtending {V}aliant-{V}azirani theorem
  to the probabilistic and quantum settings.
\newblock Available at arXiv.org e-Print quant-ph/0810.4840v1, 2008.

\bibitem[ACGK17]{ACGK17}
S.~Aaronson, A.~Cojocaru, A.~Gheorghiu, and E.~Kashefi.
\newblock On the implausibility of classical client blind quantum computing.
\newblock Available at arXiv.org e-Print quant-ph/1704.08482, 2017.

\bibitem[AD14]{AD14}
S.~Aaronson and A.~Drucker.
\newblock A full characterization of quantum advice.
\newblock {\em {SIAM} Journal on Computing}, 43(3):1131--1183, 2014.

\bibitem[Aha03]{A03}
D.~Aharonov.
\newblock A simple proof that {T}offoli and {H}adamard are quantum universal.
\newblock Available at arXiv.org e-Print quant-ph/0301040, January 2003.

\bibitem[AKN98]{AKN98}
D.~Aharonov, A.~Kitaev, and N.~Nisan.
\newblock Quantum circuits with mixed states.
\newblock In {\em Proceedings of 13th ACM Symposium on Theory of Computing
  ({STOC} 1998)}, pages 20--30, 1998.

\bibitem[Amb14]{A14}
A.~Ambainis.
\newblock On physical problems that are slightly more difficult than {QMA}.
\newblock In {\em Proceedings of 29th IEEE Conference on Computational
  Complexity (CCC 2014)}, pages 32--43, 2014.

\bibitem[AN02]{AN02}
D.~Aharonov and T.~Naveh.
\newblock Quantum {NP} - {A} survey.
\newblock Available at arXiv.org e-Print quant-ph/0210077v1, 2002.

\bibitem[AW93]{AW93}
E.~W. Allender and K.~W. Wagner.
\newblock {\em Counting hierarchies: {P}olynomial time and constant depth
  circuits}, pages 469--483.
\newblock World Scientific, 1993.

\bibitem[Bei10]{Bei08}
S.~Beigi.
\newblock {NP} vs {$\text{QMA}_{\log}(2)$}.
\newblock {\em Quantum Information and Computation}, 10:0141--0151, 2010.

\bibitem[BT09]{BT10}
H.~Blier and A.~Tapp.
\newblock All languages in {NP} have very short quantum proofs.
\newblock In {\em Proceedings of the 3rd International Conference on Quantum,
  Nano and Micro Technologies}, pages 34--37, 2009.

\bibitem[BV04]{BV04}
S.~Boyd and L.~Vandenberghe.
\newblock {\em Convex Optimization}.
\newblock Cambridge University Press, 2004.

\bibitem[DN06]{DN06}
C.~M. Dawson and M.~A. Nielsen.
\newblock The solovay-kitaev algorithm.
\newblock {\em Quantum Information and Computation}, 6(1), 2006.

\bibitem[FFKL03]{FFKL03}
S.~Fenner, L.~Fortnow, S.~A. Kurtz, and L.~Li.
\newblock An oracle builder's toolkit.
\newblock {\em Information and Computation}, 182(2):95 -- 136, 2003.

\bibitem[For94]{F94}
L.~Fortnow.
\newblock The role of relativization in complexity theory.
\newblock {\em Bulletin of the European Association for Theoretical Computer
  Science}, 52:52--229, 1994.

\bibitem[FR99]{FR99}
L.~Fortnow and J.~Rogers.
\newblock Complexity limitations on quantum computation.
\newblock {\em Journal of Computer and System Sciences}, 59(2):240--252, 1999.

\bibitem[GK12]{GK12}
S.~Gharibian and J.~Kempe.
\newblock Hardness of approximation for quantum problems.
\newblock In {\em Proceedings of 39th International Colloquium on Automata,
  Languages and Programming ({ICALP} 2012)}, pages 387--398, 2012.

\bibitem[GKS16]{GKS16}
A.~B. Grilo, I.~Kerenidis, and J.~Sikora.
\newblock {QMA} with subset state witnesses.
\newblock {\em Chicago Journal of Theoretical Computer Science}, 4, 2016.

\bibitem[GLS93]{GLS93}
M.~Gr\"{o}tschel, L.~Lov\'{a}sz, and A.~Schrijver.
\newblock {\em Geometric Algorithms and Combinatorial Optimization}.
\newblock Springer-Verlag, 1993.

\bibitem[Gol06]{G06}
O.~Goldreich.
\newblock On promise problems: {A} survey.
\newblock {\em Theoretical Computer Science}, 3895:254--290, 2006.

\bibitem[GY18]{GY16}
S.~Gharibian and J.~Yirka.
\newblock The complexity of simulating local measurements on quantum systems.
\newblock In Mark~M. Wilde, editor, {\em 12th Conference on the Theory of
  Quantum Computation, Communication and Cryptography (TQC 2017)}, volume~73 of
  {\em Leibniz International Proceedings in Informatics (LIPIcs)}, pages
  2:1--2:17, Dagstuhl, Germany, 2018. Schloss Dagstuhl--Leibniz-Zentrum fuer
  Informatik.

\bibitem[HJ90]{HJ90}
R.~A. Horn and C.~H. Johnson.
\newblock {\em Matrix Analysis}.
\newblock Cambridge University Press, 1990.

\bibitem[JKNN12]{JKNN12}
S.~P. Jordan, H.~Kobayashi, D.~Nagaj, and H.~Nishimura.
\newblock Achieving perfect completeness in classical-witness quantum
  {M}erlin-{A}rthur proof systems.
\newblock {\em Quantum Information \& Computation}, 12(5 \& 6):461--471, 2012.

\bibitem[JW09]{QRGone}
Rahul Jain and John Watrous.
\newblock Parallel approximation of non-interactive zero-sum quantum games.
\newblock In {\em {Proceedings of the 24th Annual IEEE Conference on
  Computational Complexity}}, pages 243--253, 2009.

\bibitem[KL80]{KL80}
R.~M. Karp and R.~J. Lipton.
\newblock Some connections between nonuniform and uniform complexity classes.
\newblock In {\em Proceedings of the Twelfth Annual ACM Symposium on Theory of
  Computing}, STOC '80, pages 302--309, New York, NY, USA, 1980. ACM.

\bibitem[KlGN15]{KlGN12}
H.~Kobayashi, F.~le~Gall, and H.~Nishimura.
\newblock Stronger methods of making quantum interactive proofs perfectly
  complete.
\newblock {\em {SIAM} Journal on Computing}, 44(2):243--289, 2015.

\bibitem[KSV02]{KSV02}
A.~Kitaev, A.~Shen, and M.~Vyalyi.
\newblock {\em Classical and Quantum Computation}.
\newblock American Mathematical Society, 2002.

\bibitem[KW00]{KW00}
A.~Kitaev and J.~Watrous.
\newblock Parallelization, amplification, and exponential time simulation of
  quantum interactive proof systems.
\newblock In {\em Proceedings of the 32nd ACM Symposium on Theory of Computing
  (STOC 2000)}, pages 608--617, 2000.

\bibitem[LCV07]{LCV07}
Y.-K. Liu, M.~Christandl, and F.~Verstraete.
\newblock Quantum computational complexity of the {N}-representability problem:
  {QMA} complete.
\newblock {\em Physical Review Letters}, 98:110503, 2007.

\bibitem[LFKN92]{LFKN92}
C.~Lund, L.~Fortnow, H.~Karloff, and N.~Nisan.
\newblock Algebraic methods for interactive proof systems.
\newblock {\em Journal of the ACM}, 39(4):859--868, October 1992.

\bibitem[LGG17]{LG17}
J.~Lockhart and C.~E. Gonz{\'a}lez-Guill{\'e}n.
\newblock Quantum state isomorphism.
\newblock {\em arXiv preprint arXiv:1709.09622}, 2017.

\bibitem[MS72]{MS72}
A.~Meyer and L.~Stockmeyer.
\newblock The equivalence problem for regular expressions with squaring
  requires exponential time.
\newblock In {\em Proceedings of the 13th Symposium on Foundations of Computer
  Science}, pages 125--129, 1972.

\bibitem[MW05]{MW05}
C.~Marriott and J.~Watrous.
\newblock Quantum {A}rthur-{M}erlin games.
\newblock {\em Computational Complexity}, 14(2):122--152, 2005.

\bibitem[Per12]{attila}
A.~Pereszl\'{e}nyi.
\newblock Multi-prover quantum {M}erlin-{A}rthur proof systems with small gap.
\newblock Available at arXiv.org e-Print quant-ph/1205.2761, 2012.

\bibitem[Shi03]{S03}
Y.~Shi.
\newblock Both {T}offoli and controlled-{NOT} need little help to do universal
  quantum computing.
\newblock {\em Quantum Information and Computation}, 3(1):84--92, January 2003.

\bibitem[Tod91]{T91}
S.~Toda.
\newblock {PP} is as hard as the {P}olynomial-{T}ime {H}ierarchy.
\newblock {\em SIAM Journal on Computing}, 20:865--877, 1991.

\bibitem[Vin05]{Vin05}
N.~V. Vinodchandran.
\newblock A note on the circuit complexity of pp.
\newblock {\em Theoretical Computer Science}, 347(1-2):415--418, November 2005.

\bibitem[vN28]{vN28}
J.~von Neumann.
\newblock Zur {T}heorie der {G}esellschaftspiele.
\newblock 100(1):295--320, 1928.

\bibitem[VV86]{VV86}
L.~G. Valiant and V.~V. Vazirani.
\newblock {NP} is as easy as detecting unique solutions.
\newblock {\em Theoretical Computer Science}, 47:85--93, 1986.

\bibitem[Vya03]{Vy03}
M.~Vyalyi.
\newblock {QMA$=$PP implies that {PP} contains {PH}.}
\newblock {\em Electronic Colloquium on Computational Complexity}, 2003.

\bibitem[Wat09a]{W09_2}
J.~Watrous.
\newblock {\em Encyclopedia of Complexity and System Science}, chapter Quantum
  Computational Complexity.
\newblock Springer, 2009.

\bibitem[Wat09b]{W09}
J.~Watrous.
\newblock Semidefinite programs for completely bounded norms.
\newblock {\em Theory of Computing}, 5:217--238, 2009.

\bibitem[Wat09c]{W09_3}
J.~Watrous.
\newblock Zero-knowledge against quantum attacks.
\newblock {\em SIAM Journal on Computing}, 39(1):25--58, 2009.

\bibitem[Wra76]{W76}
C.~Wrathall.
\newblock Complete sets and the {P}olynomial-{T}ime {H}ierarchy.
\newblock {\em Theoretical Computer Science}, 3(1):23 -- 33, 1976.

\bibitem[Yam02]{Y02}
T.~Yamakami.
\newblock Quantum {NP} and a quantum hierarchy.
\newblock In {\em Proceedings of the 2nd IFIP International Conference on
  Theoretical Computer Science}, pages 323--336. Kluwer Academic Publishers,
  2002.

\end{thebibliography}

\end{document}